\documentclass[12pt]{article}
\usepackage[english]{babel}    
\usepackage[ansinew]{inputenc} 
\usepackage[T1]{fontenc}       
\usepackage{lmodern,microtype} 
\usepackage{amsmath,amsthm,amssymb,amsfonts} 
\usepackage{dsfont,mathrsfs,ushort} 

\usepackage{titlesec,titling}  
\usepackage[nohead]{geometry}  
\usepackage{setspace,caption}  
\usepackage{enumitem,booktabs} 
\usepackage[comma]{natbib} 
\usepackage{tikz}
\usepackage{pstricks,pst-all}   
\usepackage{pst-plot,pst-node,pst-3dplot}
\usepackage{sgamevar}           
\usepackage{hyperref}
\usepackage{subcaption}
\usepackage{lscape}
\usepackage{pifont}
\usepackage{longtable}
\usepackage{algorithmic}
\usepackage{pdflscape}
\usepackage{rotating}
\usepackage{float}

\setenumerate{label=\small(\roman*)}
\hypersetup{colorlinks=true,pdfnewwindow=true,pdfstartview=FitH,%
	pdftitle="RCS",pdfauthor="Aguiar Boccardi Kashaev Kim",%
	urlcolor=blue!90!red!45!black,citecolor=blue!90!red!45!black,linkcolor=red!90!black}
\DeclareCaptionFont{fancy}{\bfseries\sffamily}
\gamemathtrue
\allowdisplaybreaks

\providecommand{\psreset}{\psset{%
		linewidth=0.3pt,linestyle=solid,linecolor=black,
		dotsize=2.5pt,dotsep=2.5pt,arrowsize=4pt,
		fillstyle=none,fillcolor=white,
		showpoints=false,arrows=-,linearc=0,framearc=0,
		hatchsep=2pt,hatchwidth=0.2pt,nodesep=4pt,opacity=1}
	\psset{gridcolor=black!60, subgridcolor=black!30}
}

\psreset

\usepackage{graphicx}
\graphicspath{ {figures/} }

\titleformat{\section}[block]{\centering\large\bfseries\sffamily}{\thesection.}{0.5em}{}
\titleformat{\subsection}[block]{\flushleft\bfseries}{\thesubsection.}{0.5em}{}
\titleformat{\subsection}[block]{\flushleft\bfseries\sffamily}{\thesubsection.}{0.5em}{}
\titleformat{\subsubsection}[runin]{\normalsize\bfseries\sffamily}{\bfseries\upshape\sffamily\thesubsubsection.}{0.5em}{}[.--\:]
\renewcommand{\thesubsubsection}{\arabic{section}.\arabic{subsection}.\arabic{subsubsection}}
\titlespacing{\section}{0ex}{10ex}{5ex}
\titlespacing{\subsection}{0in}{6ex}{3ex}
\titlespacing{\subsubsection}{0mm}{2ex}{0.5em}
\pretitle{\begin{center}\LARGE\bfseries\sffamily}
\posttitle{\par\end{center}\vskip 0.5em}
\preauthor{\begin{center} \large \lineskip 0.5em\begin{tabular}[t]{c}}
\postauthor{\end{tabular}\par\end{center}}
\predate{\begin{center}\small}
\postdate{\par\end{center}}
\providecommand{\abstitle}[1]{{\par\vspace*{2ex}\small\bfseries\sffamily #1}\hspace*{1ex}}
\renewenvironment{abstract}%
{\begin{center}\begin{minipage}{0.8\linewidth}%
			\setlength{\parindent}{0.0em}\abstitle{Abstract}\small}%
		{\end{minipage}\end{center}\vfill\clearpage}

\DeclareMathOperator*{\argmax}{arg\,max}

\providecommand{\Char}[1]{\mathds{1}\left(\,#1\,\right)}
\providecommand{\Real}{{\mathds{R}}}

\providecommand{\tr}{^{\prime}}

\providecommand{\rand}[1]{\mathbf{#1}}
\providecommand{\rands}[1]{\boldsymbol{#1}}
\providecommand{\norm}[1]{\left\lVert#1\right\rVert}
\providecommand{\Prob}[1]{\mathds{P}\left(#1\right)}
\providecommand{\Exp}[1]{\mathds{E}\left[#1\right]}
\providecommand{\abs}[1]{\left\lvert#1\right\rvert}

\newcommand{\A}{A} 
\newcommand{\D}{D}
\newcommand{\M}{\mathcal{\A}}

\newcommand{\rr}{\succ}
\newcommand{\rc}{\ensuremath{\mathrm{L}}}
\newcommand{\mm}{\ensuremath{\mathrm{MM}}}
\newcommand{\la}{\ensuremath{\mathrm{LA}}}
\newcommand{\rcg}{\ensuremath{\mathrm{EBA}}}
\newcommand{\rum}{\ensuremath{\mathrm{FC}}}
\newcommand{\rumo}{\ensuremath{\mathrm{RUM}}}
\newcommand{\hrc}{\ensuremath{\mathrm{B}}}
\newcommand{\rhrc}{\ensuremath{\mathrm{L\text{-}B}}}

  \theoremstyle{remark}
  
  \theoremstyle{plain}
  \newtheorem{lemma}{\protect\lemmaname}
  \theoremstyle{definition}
  \newtheorem{definition}{\protect\definitionname}
\theoremstyle{plain}
\newtheorem{theorem}{\protect\theoremname}

  \theoremstyle{plain}
  \newtheorem{cor}{\protect\corollaryname}
 \theoremstyle{definition}
  \newtheorem{example}{\protect\examplename}
  \theoremstyle{plain}
  
  \providecommand{\assumptionname}{Assumption}

  \providecommand{\definitionname}{Definition}
  \providecommand{\lemmaname}{Lemma}
  
  \providecommand{\remarkname}{Remark}
\providecommand{\corollaryname}{Corollary}
\providecommand{\theoremname}{Theorem}
\providecommand{\examplename}{Example}

\makeatother

\begin{document}
\title{Random Utility and Limited Consideration\thanks{\scriptsize This paper was previously circulated as ``Does Random Consideration Explain Behavior when Choice is Hard? Evidence from a Large-scale Experiment.'' We are grateful to an editor and 4 anonymous referees for insightful comments and suggestions. We would like to thank Roy Allen, Jose Apesteguia, Miguel Ballester, Levon Barseghyan, Juan Dubra, Mikhail Freer, Yoram Halevy, Yuichi Kitamura, Paola Manzini, Marco Mariotti, John Rehbeck, and J\"{o}rg Stoye for useful comments and suggestions. We also thank the participants of the Barcelona GSE Summer Forum (Stochastic Choice), BRIC 2019, IWABE 2019, ASSA 2020 (Econometrics of Decision and Demand) for useful feedback. Mingshi Kang provided excellent research assistance. Aguiar and Kashaev gratefully acknowledge financial support from the Western Social Science Faculty grant and Social Sciences and Humanities Research Council. } \thanks{This study is approved by the Caltech's IRB No. 18-0812.}}
\author{Victor H. Aguiar\thanks{Department of Economics, University of Western Ontario, vaguiar@uwo.ca}\quad
Maria Jose Boccardi\thanks{Amazon, majoboccardi@gmail.com.  The experiment was run prior Boccardi joined Amazon}\quad
Nail Kashaev\thanks{Department of Economics, University of Western Ontario, nkashaev@uwo.ca} \quad
Jeongbin Kim\thanks{Department of Marketing, National University of Singapore, bizkj@nus.edu.sg} }

\date{This version: June 2022/ First version: November 2018}
\maketitle
\begin{abstract} 
The random utility model  \citep[RUM,][]{mcfadden1990stochastic} has been the standard tool to describe the behavior of a population of decision makers. RUM assumes that decision makers behave as if they maximize a rational preference over a choice set. This assumption may fail when consideration of all alternatives is costly. We provide a theoretical and statistical framework that unifies well-known models of random (limited) consideration and generalizes them to allow for preference heterogeneity. We apply this methodology in a novel stochastic choice dataset that we collected in a large-scale online experiment. Our dataset is unique since it exhibits both choice set and (attention) frame variation. We run a statistical survival race between competing models of random consideration and RUM. We find that RUM cannot explain the population behavior. In contrast, we cannot reject the hypothesis that decision makers behave according to the logit attention model \citep{brady2016menu}.\\


\end{abstract}

\section{Introduction}\label{section:Introduction}

A fundamental question in social science is how to describe the behavior of a population of decision makers (DMs).  The random utility model \citep[RUM,][]{mcfadden1990stochastic} is the standard tool to \emph{describe} behavior.\footnote{$\rumo$ was first proposed by \citet{block_random_1960} and \citet{falmagne_representation_1978} in an environment similar to ours.} $\rumo$ assumes that DMs behave \emph{as if} they maximize their preferences over their choice set. 
However, $\rumo$ may fail at describing behavior  if DMs do not consider all available alternatives.  For instance, DMs may behave as if they use a two-stage procedure: first simplifying choice by using a consideration set, and only then choosing the best alternative among those considered. Thus, DMs may choose dominated alternatives.\footnote{For evidence of choice of dominated alternatives see \citet{santos,honka2014quantifying,  heiss2016inattention, ho2017impact, honka2017advertising, hortaccsu2017power} and \citet{molinari2018}.} 
A large literature, pioneered by \citet{masatliogluRA} and \citet{manzini2014stochastic}, has proposed theories of consideration-mediated choice. These theories accommodate departures from $\rumo$ caused by inattention, feasibility, categorization, and search.\footnote{See, for instance, \citet{ABDsatisficing16,brady2016menu,aguiar2017random, lleras2017more, caplin2016rationalconsideration,horan2018,kovach2017satisficing}, and \citet{cattaneo2017random}.}
In contrast to $\rumo$,\footnote{For evidence for $\rumo$  see \citet{kitamura2018nonparametric} and \citet{McCauslandExp12017}.} little is known about the empirical validity of these models. Our work aims to fill this important gap in the literature. 
\par
Methodologically, we provide a unifying theoretical framework that generalizes well-known theories of random consideration. We unify these theories with a new concept called attention-index. The attention-index is a net measure of how enticing a collection of alternatives is, or how costly it is to pay attention to it.  We show how to test these theories statistically, and how to recover the preference distribution and consideration rules. Our framework extends many theories of consideration-mediated choice to allow for preference heterogeneity. This allows us to take these theories of individual behavior to the population level, thus permitting the use of cross-sectional datasets to test them. 
Following \citet{mcfadden1990stochastic} and \citet{kitamura2018nonparametric}, we take seriously the fact that all theories of stochastic choice have as their primitive the unobserved distribution over choices that can only be estimated in finite samples by sample frequencies of choice.  Thus, to test these models in finite samples, we need to account for sampling variability.
\par
Empirically, we design a large experiment\footnote{Other experiments that we are aware of that have collected stochastic choice data focusing on choice set variation are \citet{apesteguia2018separating} ($87$ individuals) and \citet{McCauslandExp12017} ($141$ participants). In contrast to our work, both focus mainly on binary choice sets and goodness-of-fit measures (including the computation of Bayes factors).} with two independent sources of exogenous variation: (i) full variation in choice sets (menus),  and (ii) variation in frames. We conducted this experiment online in Amazon Mechanical Turk (MTurk), collecting 12297 independent choice observations from 2135 individuals. A \emph{frame} consists of  observable information that is irrelevant in the rational assessment of the alternatives \citep{salant2008f}.\footnote{We interpret the notion of rational assessment to be an assessment compatible with consequentialism and $\rumo$.}   
Full variation in choice sets means that all possible choice sets are observed by the researcher. It allows us to test consideration-mediated choice theories in a large cross-section of heterogeneous individuals. Frame variation means that we vary the complexity of the description of alternatives without affecting the relevant payoffs. This is equivalent to varying the cost of consideration. In this sense, we induce an attention frame. This variation in frames allows us to differentiate between $\rumo$ and models of limited consideration since consideration could change with frames but preferences must remain stable. 
\par
In general, without frame variation, many models of consideration are empirically indistinguishable from $\rumo$. For instance, if a dataset, without frame variation, is consistent with the classical consideration model of \citet{manzini2014stochastic}, then it is also consistent with $\rumo$. However, this model and $\rumo$ will typically recover a distinct distribution of preferences. Varying frames, we can test whether the distribution of preferences remains the same across frames. This will imply that the model in \citet{manzini2014stochastic} and $\rumo$ may have different empirical implications in our experimental setup. The same reasoning applies to any model of limited consideration.  
\par
In our experimental design, we introduce three attention-cost/complexity frames for every choice set. Each DM faces all three frames. These treatments require the DM to solve a simple cognitive task to understand the alternatives. The consideration cost is progressively reduced across frames while we keep the choice set fixed. The fact that the choice remains the same is not explicitly stated in the experiment instructions. Within this design we can understand how changing consideration costs across frames may affect choices while we keep (the distribution of) preferences fixed. Under our incentives protocol (pay-at-random across tasks), and under consequentialism, the distribution of preferences must remain constant regardless of the frame. By exploiting this feature of our design, we show that, in our sample, $\rumo$ fails to describe the population behavior, while the logit attention model of \citet{brady2016menu} describes it well. 
\par
Our paper generalizes the methodology of \citet{cattaneo2017random} by allowing heterogeneity in preferences when the choice sets include a default alternative. We show that in our setting, if one assumes independence of preferences and consideration, at the population level, testing a consideration-mediated choice theory with heterogeneous preferences and a dominated default is equivalent to testing whether this hypothetical full-consideration distribution over choices is consistent with $\rumo$.
We, therefore, use this result to test these models using the framework of \citet{kitamura2018nonparametric}. 
\par
To exploit all possible implications of the limited-consideration models of interest, we need full variation in choice sets. A limited consideration model may describe behavior well for a nonexhaustive dataset, but it may fail to do so for an extended one. That is, one may have false positives when observing choices from a nonexhaustive set of menus, as discussed at length in \citet{declippel2018} and \citet{cattaneo2017random}. Nonetheless, full exogenous variation in choice sets is an important data feature that is usually not satisfied in field data. Our testing procedure exploits rich experimental variation to increase the statistical power of our tests. 
\par 
In addition, our experiment introduces a dominated default alternative that works as an opportunity cost of paying attention. For any choice set and frame, the default is always present and shown first. Moreover, it is pre-selected as the default choice--if the subject decides to skip the task, she is informed that the default alternative will be chosen for her. This design allows us to use the default alternative as the opportunity cost of incurring in the cost of consideration and understanding the other alternatives in the choice set. The set of alternatives in our experiment consists of lotteries. Hence, we use a degenerate lottery as the default due to its simplicity. In this sense, we believe the default alternative in our design has effectively zero cost of consideration. This dominated default is key to disentangle the distribution of preferences from random consideration.\footnote{We formulate a sensitivity analysis when the default is not dominated in Appendix~\ref{appendix:sensitivitydefault}.} 
\par
We use these theoretical and experimental innovations to test two well-known models of random consideration: (i) the logit attention model of \citet{brady2016menu} ($\la$), and (ii) the version of elimination-by-aspects model of \citet{tversky1972elimination} characterized by \citet{aguiar2017random} ($\rcg$).\footnote{We refer to $\rcg$ as the version of the original model in \citet{tversky1972elimination} by \cite{cattaneo2017random} (Example $6$). The model in \citet{aguiar2017random} coincides with $\rcg$ in the special case where there is a dominated default with a restriction that any category that does not contain the default has zero mass. The $\rcg$ model characterization for the case without a default remains an open question.} There is increasing interest in incorporating limited consideration in discrete choice.\footnote{See, for instance, \citet{goeree2008limited,molinari2018,barseghyan2019discrete,fieldMM19}, and \citet{abaluck2017consumers}.} In particular, the influential and tractable model of \citet{manzini2014stochastic} ($\mm$) has become an important tool for the analysis of limited attention in empirical work (e.g., \citealp{fieldMM19,abaluck2017consumers}, and \citealp{kashaev2019peer}). However, $\mm$ is highly stylized and assumes that consideration is driven by an item-dependent parameter (i.e., independence in consideration). We investigate from an experimental perspective whether this strong assumption is effective in explaining choice when consideration is hard. To do so we consider two extensions of $\mm$ that allow for substitution and complementarity in consideration, the $\la$ and $\rcg$ models. These two generalizations have the property that their intersection is exactly the $\mm$ model \citep{suleymanov}. This implies that $\mm$ explains the population behavior if and only if both the $\la$ and $\rcg$ models explain it. 
\par
We test these models and the benchmark $\rumo$ conditioning on the frame. Crucially, we require the underlying preference relation to be stable among frames while allowing the consideration rules to vary with the frame. Our main findings are: (i) We reject the hypothesis that $\rumo$ provides a good description of population behavior. (ii) In contrast, the $\la$ model with heterogeneous preferences cannot be rejected at the $5$ percent significance level. (iii) However, we reject the hypothesis that $\rcg$, and hence $\mm$, describe the population behavior. \par
Our work contributes to the recent experimental literature on stochastic choice, limited consideration, and departures from  $\rumo$. Even though by now, limited attention in many environments is well documented, it is less clear what structural models of limited attention should the practitioner use. We see our contribution as providing answers to this second issue.\footnote{There is a vast literature documenting departures from fully rational behavior, but it is not focused on limited consideration. See \citet{rieskamp2006extending} for a survey.} We hope that our findings about which models of limited consideration are successful empirically will inform future empirical work in the field. For instance, our findings have already been used to motivate the choice of the parametric specification of limited consideration in the recent work of \citet{abaluck2017consumers}.\footnote{\cite{abaluck2017consumers} structurally estimate  a model of discrete choice with limited consideration under the $\la$ and $\mm$ model.} 

\subsection*{Outline } The paper proceeds as follows. 
Section~\ref{section:model} presents our model.  
Section~\ref{section:frames} details frame variation and our testing procedure. Section~\ref{section:experiment} presents our experiment.
Section~\ref{section:testing} presents the testing results. Finally, Section~\ref{section:conclusion} concludes. All proofs and additional results are in the appendix.

\section{Environment -- Model}\label{section:model}
We consider a finite choice set $X$ and we denote the outside alternative or default as $o\notin X$. We let the set of all possible choice sets be $\M=2^{X}\setminus\{\emptyset\}$, where $2^{X}$ denotes the set of all subsets of $X$. A probabilistic choice rule is a mapping $p:X\cup\{o\}\times\M\mapsto[0,1]$. The probabilistic choice rules for a given choice set add up to 1, $\sum_{a\in A}p(a,A)+p(o,A)=1$. Moreover, $p(a,A)=0$ if $a\notin A$. We fix $p(o,\emptyset)=1$. A complete stochastic choice rule is a vector $P=\left(p(a,A)\right)_{A\in\M,a\in A\cup\{o\}}$. For identification purposes, we treat $P$ as a known object. In practice, we do not observe $P$, but can consistently estimate it by the collection of sample frequencies $\hat{P}$ (see  Section~\ref{subsection: test} for further details). 

\subsection{Random Behavioral Model}\label{subsection:rhrc}
We consider an environment where DMs, faced with a choice set $\A\in\M$, first pick $D\subseteq A$ (consideration set) and then choose the alternative in $D$ that maximizes their preferences. With probability $\pi(\succ)$, DMs are endowed with preferences $\succ \in X\times X$ drawn from the set of all linear orders (strict preferences) on $X$, $R(X)$.\footnote{Linear orders are complete, reflexive, transitive, and antisymmetric orders.} Note that since $o\not\in X$, following \citet{manzini2014stochastic}, we implicitly assume that the default is picked if and only if nothing else is considered. A typical interpretation of this situation is the sleeping agent behavior (see, for instance, \citealp{abaluck2017consumers}). When the agent is sleeping (i.e., she considers the empty set) she chooses the default alternative. Otherwise, the agent wakes up and considers some nonempty set and maximizes her preferences in her consideration set. In the second case, the default is assumed to be dominated by the rest of alternatives. We show how to relax this assumption in Appendix~\ref{appendix:sensitivitydefault}.
\par
The distribution $\pi\in\Delta(R(X))$ fully captures preference heterogeneity.\footnote{For any set $C$, $\Delta(C)$ denotes the set of all probability distributions (simplex) on $C$. $\Char{B}$ denotes indicator of the statement $B$ and is equal to $1$ if $B$ is true and is equal to zero otherwise.} The distribution over random consideration sets given the menu $A$ is fully characterized by $m_\A:2^{\A}\to[0,1]$, $\sum_{\D\subseteq \A}m_\A(\D)=1$. In other words, $m_\A$ is an element of the simplex $\Delta\left(2^\A\right)$. Let $m$ denote the complete collection of those distributions for all possible menus. That is, $m=\left(m_A(D)\right)_{A\in\M,D\in 2^{A}}$. We assume that the random consideration sets and random preferences are \emph{independent}.

\begin{definition} [Random Behavioral Model, $\hrc$-rule] \label{def:HRC-rule}
A complete stochastic choice rule $P$ is a $\hrc$-rule if there exists a pair $(m,\pi)$ such that 
\[
p(a,\A)=\sum_{\rr\in R(X)}\pi(\rr)\sum_{\D\subseteq \A}m_\A(\D)\Char{a \rr b,\:\forall\: b\in \D}
\]
for all $a\in X$ and $\A\in\M$.
\end{definition}

\begin{figure}[t]
\centering
\begin{tikzpicture}	[scale=1.05]
\draw [line width=1.25pt, fill=blue!10](0,5) ellipse (1.25cm and 0.75cm);
\draw (0,4.75) node[above] {\footnotesize{$A$}};
\draw (0,3.5) node[above] {$A$ : \textsc{Choice Set}};
\draw [line width=1.25pt, fill=blue!10](5,5) ellipse (1.5cm and 0.75cm);
\draw [dotted, line width=1.5pt, fill=red!25 ] (4.5,4.5) -- (4.5,5.5) -- (5.5,5.5)--(5.5,4.5) --(4.5,4.5);
\draw (5,4.75) node[above] {\footnotesize{$D$}};
\draw (5,3.5) node[above] {$D$ : \textsc{Consideration Set}};
\draw (2.5,6.25) node[right] { $m_A(\cdot)$};
\draw (7.75,6.25) node[right] { $\pi(\cdot)$};
\draw (7.5,7.6) node[above] {\textsc{Preferences: }};
\draw (7.5,7) node[above] {$\succ\in R(X)$};
\draw (2.25,7.6) node[above] {\textsc{Consideration Sets: }};
\draw (2.25,7) node[above] {$D\in 2^A$};
\draw [thick,->,>=stealth, line width=1.5pt, dotted] (2.25,6.75)--(2.25,5.25);
\draw [thick,->,>=stealth, line width=1.5pt, dotted] (7.5,6.75)--(7.5,5.25);
\draw [thick,->,>=stealth, line width=1.5pt, dotted] (1.5,5)--(3.35,5);
\draw [thick,->,>=stealth, line width=1.5pt, dotted] (6.7,5)--(8.45,5) node[right] {$\mathbf{a_{i,A}}$ : $\succ$-maximal in $D$ };
\draw (10,3.5) node[above] {$\mathbf{a_{i,A}}$ : \textsc{Choice}};
\end{tikzpicture}	
	\caption{\textsc{Consideration mediated choices } Choices are the result of a two stage process, first pick a consideration set and then pick the best alternative in that set. We observe choices ($\mathbf{a_i}$) and menus ($A$). We do not observe and we aim to identify the theory objects: distribution of preferences in the population ($\pi$) and stochastic choice rule ($m_A$).} \label{figure:choice}
\end{figure}

This choice rule is illustrated in Figure~\ref{figure:choice}. Definition~\ref{def:HRC-rule} implicitly assumes that the random consideration set rule and the heterogeneous preferences are independent. Independence is a good starting assumption in the sterile environment of our experiment, as we want to achieve a decomposition of any observed probabilistic choice rule into its consideration (captured by $m$) and preference (captured by $\pi$) components. Independence has been assumed successfully in the structural work of \citet{abaluck2017consumers}.
Also, we are interested in modeling decision-making in two-stages where DMs simplify a hard choice task using fast-and-frugal heuristics (consideration) that are independent of preferences, and then choose rationally from the simplified choice set. If the researcher observes additional information (e.g. age, gender, education, and income levels of individuals), then random consideration rule and random preferences need to be independent only conditionally on those observables.\footnote{See \citet{kashaev2021random} for a study of the correlation between preferences and consideration.}
\par
Independence holds trivially for the case of homogeneous preferences, such as all models covered by the Random Attention Model (RAM) of \citet{cattaneo2017random}. Moreover, as the following lemma demonstrates, the $\hrc$-rule does not have empirical content even under the independence assumption. 
\begin{lemma}\label{lemma:empiricalcontent}
Every complete stochastic choice rule $P$ is a $\hrc$-rule. 
\end{lemma}
Without additional restriction on $\pi$ and $m$ the model is not falsifiable. That is, any $P$ can be generated by some independent $\pi$ and $m$. We will impose constraints on $m$ that will allow us to test an important class of random consideration sets models without restricting heterogeneity in preferences.

\subsection{Attention-Index Consideration Set Rule}
We restrict $m$ by considering a family of consideration set rules that are governed by an \textit{attention-index}. The attention-index $\eta \in \Delta(2^X)$ is a distribution over the power set. The value $\eta(D)$ captures the unconditional attention that DMs pay to the set $D\in 2^X$. The attention-index of a set is a net measure of its attractiveness with respect to how hard it is to consider it. The attention-index measures how enticing a consideration set is, and how complex it is to understand. Therefore, $\eta(C)>\eta(D)$ means that $C$, in net terms, attracts more attention than $D$.

\begin{definition}[Attention-index representation]
A consideration set rule $m$ admits an attention-index representation if there exists and attention-index $\eta$, a link function $\psi$, and an index correspondence $g$ such that $g(D,A)\subseteq 2^X\setminus D$ and
\[
m_{A}(D)=\psi\left(\eta(D),\sum_{C\in g(D,A)}\eta(C)\right)
\]
for all $A\in\M$ and $D\subseteq A$. 
\end{definition}
In what follows, we assume that the link function $\psi$ and the correspondence $g$ are known. A given link function captures the particular way in which the attention-index shapes consideration given a choice set or menu. In other words, the link function transforms unconditional attention into conditional (on the choice set) attention. The index set $g(D,A)$ captures the collection of sets that are used in the attention aggregator $\sum_{C\in g(D,A)}\eta(C)$. However, we do not assume that $\eta$ is known and do no impose any restrictions on it, so the setting is still semiparametric.
\par
Next, we define several important models of limited consideration that admit the attention-index representation. They only differ in how they use the attention-index to form the consideration set at a given menu.
\begin{definition}\label{def:models}
The logit attention ($\la$, \citealp{brady2016menu}), the choice set independent ($\mm$, \citealp{manzini2014stochastic}), the random consideration ($\rcg$, \citealp{tversky1972elimination, aguiar2017random}\footnote{The model in \citet{aguiar2017random} is a special case of the model in \citet{tversky1972elimination} and coincides with it when the attention-index of only sets that contain the default is nonzero.}), and the full consideration ($\rum$) models admit an attention-index representation such that:
\begin{itemize}
\item $m\in\mathcal{M}^{\la}$ if and only if there exists $\eta\in \Delta(2^X)$ such that for
\[
m_A(D)=\frac{\eta(D)}{\sum_{C\subseteq A}\eta(C)}>0
\]
for all $\A\in\M$ and $D\in 2^{A}$;
\item $m\in\mathcal{M}^{\mm}$ if and only if $m\in\mathcal{M}^{\la}$ with 
\[
\eta(A)=\prod_{a\in X\setminus A}\left(1-\gamma(a)\right)\prod_{b\in \A}\gamma(b)
\]
for a given $\gamma:X\to (0,1)$ and for all $A\in 2^{X}$;
\item $m\in\mathcal{M}^{\rcg}$ if and only if there exists $\eta\in \Delta(2^X)$ such that 
\[
m_A(D)=\sum_{C:C\cap A=D}\eta(C)
\]
for all $\A\in\M$ and $D\in 2^{A}$;
\item $m\in\mathcal{M}^{\rum}$ if and only if $m\in \mathcal{M}^{\rcg}$ with
\[
\eta(A)=\Char{A=X}.
\]
\end{itemize}
\end{definition}
The $\mm$-rule imposes independence in consideration across items, making the model highly stylized and tractable. Both the $\la$- and the $\rcg$-rules generalize this item-independence to allow for substitution and complementarity of attention between different items. The $\la$-rule predicts that the probability of considering $D$ is proportional to the attention-index value $\eta(D)$. The $\rcg$-rule predicts that the probability of considering a subset $D$ given menu $A$ is equal to the probability that $D$ is the intersection of the subset of alternatives (considered randomly using the attention-index) and the choice set.\footnote{Note that for the $\la$-model $\psi^{\la}(x,y)=x/(x+y)$ and $g^{\la}(D,A)=\{C\in 2^X\setminus{D}\::\:C\cap A=C\}$, and for the $\rcg$-model $\psi^{\rcg}(x,y)=x+y$ and $g^{\rcg}(D,A)=\{C\in 2^X\setminus{D}\::\:C\cap A=D\}$.} In our application, we focus on these four particular models because they allow us to learn about the true data generating process governing our experimental application systematically. \footnote{$\la$ and $\rcg$ are completely distinct generalizations of $\mm$. Thus, rejecting or accepting either of these models is very informative.} However, our approach extends to any model of consideration that admits an attention-index representation.\footnote{For instance, one can generate a continuum of such models by taking all possible convex combinations of consideration rules $m^\la$ and $m^\rcg$ that are generated by the same attention-index $\eta$.}
\par
Given the definition of the attention-index representation, we can define a restricted version of the $\hrc$-rule ($\rhrc$-rule).
Let $\mathcal{M}^{\rc}$ be the set of all consideration rules induced by a given link function $\psi^{\rc}$, the correspondence $g^{\rc}$, and all attention-indices $\eta$
\[
\mathcal{M}^{\rc}=\left\{m\::\:m_{A}(D)=\psi^{\rc}\left(\eta(D),\sum_{C\in g^{\rc}(D,A)}\eta(C)\right)\text{ for some }\eta\in\Delta(2^{X})\right\}.
\]
\begin{definition} [$\rhrc$-rule] \label{def:Restricted HRC} 
A complete stochastic choice rule $P$ is $\rhrc$-rule if $P$ is a $\hrc$-rule with $m\in\mathcal{M}^{\rc}$.
\end{definition}
We believe that random consideration rules that admit an attention-index representation have several theoretical advantages over other generalizations in the literature.  
First, they cover, as special cases, well-known models of consideration and allow us to introduce a new class of semiparametric models of limited consideration with heterogeneous preferences. In other words, the $\rhrc$-rule unifies existing models into a common structure. This structure helps us to understand the common traits of these models, to classify them, and to provide an identification and testing framework to possibly new models covered in this structure. 
\par
Second, they enable the unique identification of the consideration rule and the underlying stochastic choice full consideration probability from a cross-section of choices with menu variation and heterogeneous preference (see Theorem~\ref{thm:identification}). We are not aware of any other work that achieves the same and is more general than ours.
Allowing heterogeneous preferences is important in analysis of any datasets, and is unavoidable in cross-sections. An important alternative generalization is RAM, which imposes only a shape constraint on random consideration. However, it does not allow heterogeneous preferences, nor it obtains point identification of consideration or the underlying full consideration choice rule. Moreover, RAM is completely uninformative about preferences, when the stochastic choice rule is regular,\footnote{Regular random choice means that $p(a,A)\geq p(a,B)$ for any two menus $A\subseteq B$ and any $a\in A$.} as in the random utility models.\footnote{\citet{kashaev2021random} extend RAM to allow for heterogeneous preferences and  show that when the stochastic choice rule is regular nothing can be said as well about random consideration and preferences even under independence among them.}  Also, the $\rhrc$-rule allows cycles of probabilities that can violate regularity while RAM cannot (see Appendix~\ref{subsectionapp:comparison_RAM}). 
\par
Third, the concept of attention-index is intuitive, simple, and of behavioral interest, as it provides a useful (unconditional) index of attention for any subset of alternatives. The link function and the attention-index follow the tradition of classical stochastic choice theory of simple scalability where the probability of choice of an alternative in a menu is a nonlinear function (a link function) of some scale/index \citep{krantz1965scaling,tversky1972elimination}. In the simple scalability tradition, the scale/index captures the intensity of the stimuli associated with a particular alternative. A classical example is the logit model of choice. Here, we apply the same intuition to consideration--the attention-index captures the net attractiveness of a consideration set.
\par
Fourth, consideration rules that admit an attention-index representation are compatible with optimal random consideration of a representative DM.
Here, we show that a $\rhrc$-rule can be obtained as the result of allocating attention optimally (see Example~\ref{ex: costly attention}).
\par
In addition, consideration rules that admit an attention-index representation are compatible with a flexible interpretation of randomness. In our framework, the randomness due to limited consideration can arise both at the individual and population level. Indeed, consideration can be random at the individual level and independent and identically distributed (i.i.d.) at the population level; and consideration can be deterministic at the individual level but heterogeneous at the population level.
The next example shows how the latter case can be described by a $\rhrc$-rule.
\begin{example}[Heterogeneous Categorization]
Consider a population of DMs with two types of agents endowed with different deterministic attention rules. Assume that half of the DMs are fully attentive, while the other half follows a rule of thumb: they pay full attention to option $b$ if it is present in a given choice set, else the consideration set is empty. The DMs pick the best alternative, according to the (heterogeneous) preference realization governed by some $\pi\in \Delta(R(X))$. If the consideration set is empty, then the outside option is selected. This population has heterogeneous (deterministic) consideration, however, the population behavior can be fully captured by a random consideration rule with the $\rcg$ restriction. Namely, let $\eta(X)=\frac{1}{2}$ and $\eta(\{b\})=\frac{1}{2}$. Then the $\rcg\text{-}\hrc$-rule can describe this population behavior.
\end{example}

The next example demonstrates that the consideration rules that admit an attention-index representation can be derived as a solution to the problem where representative DMs optimally allocate their attention.
\begin{example}[Costly Attention Allocation]\label{ex: costly attention}
Consider a (representative) DM  whose preferences are given by a (mean) utility $u:X\to\Real$ and additive random vector $\rands{\xi}=(\rands{\xi}_x)_{x\in X}$ such that the random utility of a given item $x$ is given by $u(x)+\rands{\xi}_x$. The taste shocks $\rands{\xi}$ are distributed with respect to some continuous distribution.\footnote{This guarantees that the implied random utility rule $\pi$ exists because ties have zero probability.}
When the DM is faced with a menu $A$, she needs to allocate her attention, measured by $m_A\in \Delta(2^A)$, over all possible consideration sets in $A$ (including the empty set). The attractiveness of a set $D$ is captured by the McFadden's surplus of a given set defined by $\alpha(D)=\Exp{\max_{x\in D} \left(u(x)+\rands{\xi}_x\right)}$ for all $D\subseteq 2^A\setminus\{\emptyset\}$, where the expectation is taken with respect to $\rands{\xi}$. The attractiveness of the empty set is normalized to be $0$, $\alpha(\emptyset)=0$. The surplus $\alpha(D)$ is a measure of average attractiveness, capturing how enticing a consideration set is for the representative DM. 
The difficulty of picking a consideration set, or the cost of attention, is captured by a cognitive cost function $K:[0,1]\to\Real\cup\{\infty\}$. If $D$ is considered with probability $m(D)$, then the cognitive cost is $K(m(D))$. The cost function is menu independent, but depends on the allocated attention $m(D)$. Following \citet{fudenberg2015stochastic}, we assume that $K$ is convex. In this case, DM's problem is to find $m_A\in \Delta(2^A)$ that maximizes the expected attractiveness of the menu, given the cognitive cost of processing it. Formally,
\[
m_A=\argmax_{m\in \Delta(2^A)} \sum_{D\subseteq A}[m(D)\alpha(D)- K(m(D)) ].
\]
When $K(t)=0$ for all $t\in[0,1]$, the solution is such that $m_A(A)=1$. That is, the DM is consistent with full consideration ($\rum$). When $K(t)= -t\log (t)/\theta$, where $\theta$ is the cost parameter, we get that
\[
m_A(D)=\frac{\exp{(\theta \alpha(D))}}{\sum_{C\subseteq A} \exp{(\theta \alpha(C))}}.
\]
That is, optimal consideration in this case is consistent with $\la$ with the attention-index $\eta(D)=\exp{(\theta \alpha(D))}$. Note that in this definition of $\eta$, the independence assumption is still satisfied since $\eta(D)$ is an aggregate quantity and only depends on the mean utilities $(u(x))_{x\in X}$ and the distribution of $\rands{\xi}$ but not on the realizations of $\rands{\xi}$ (e.g., when $\rands{\xi}$ follows the Gumbel distribution, then $\alpha(D)=\log \sum_{x\in D}e^{u(x)}$). Of course, we can replace $\alpha(D)$ with any other measure of attractiveness and all our derivations will go through. This means we can generate any model consistent with $\la$ this way. Finally, when $K(t)=\frac{1}{2} t^2$, we can get a special case of the $\rcg$ model.\footnote{The $\rcg$ model then requires that the DM does not re-optimize in smaller menus $A\subset X$. Instead, she uses the heuristic that whatever category she draws at $X$ it is intersected with the given menu to obtain the consideration set.} 
\end{example}
\par
Finally, the attention-index representation also has several econometric
advantages. It allows for a significant reduction of the dimensionality of the consideration set rules, making them tractable. In general, the number of parameters controlling the random consideration is $\sum_{A\subseteq X}2^{\abs{A}}-1$. The single-attention index and the link function reduce the number of unknown parameters to $2^{\abs{X}}-1$. It also leads to statistical testing (i.e. we can take into account sample variability) of known models of random consideration in a cross-section dataset of choices (see Section~\ref{section:testing}). Thus, one can confront existing models to experimental datasets in a competitive fashion to guide the exploration of models of consideration sets and to inform which models are more successful.

\subsection{Characterization and Identification of the \texorpdfstring{$\rhrc$}{Link-Behavioral}-model}\label{subsection:rhrc characterization}
In this section, we answer the following questions: (i) When can we recover different consideration rules from the data? (ii) What are their observable implications? 
We answer these questions by \textit{decomposing} the observed probabilities of choice $P$ into an attention rule $m$ and a distribution of preferences $\pi$. In other words, we recover from the dataset $P$ the primitives of the $\hrc$-rule that generated it, and provide necessary and sufficient conditions that guarantee that a dataset $P$ can be generated by a $\rhrc$-rule. 
\par
Our starting point is to exploit the fact that if a consideration rule $m$ admits an attention-index representation, then the probability of choosing the default alternative is completely determined by the attention-index $\eta$. In particular, the probability of choosing the outside option is independent of the distribution of preferences due to the independence assumption that we have imposed between preferences and consideration. In addition, recall that in our model the outside option is only chosen when nothing else in the menu is considered.\footnote{For an extension that relaxes this assumption see Appendix~\ref{appendix:sensitivitydefault}.} If we denote $p_o=(p(o,A))_{A\in \M}$ and $\psi_{\emptyset}(\eta)=\left(\psi\left(\eta(\emptyset),\sum_{C\in g(\emptyset,A)}\eta(C)\right)\right)_{A\in\M}$, abusing notation we can write the system of equations 
\[
p_o=\psi_{\emptyset}(\eta).
\]
When $\psi_{\emptyset}$ is \textit{invertible}, we can uniquely recover the random consideration rule from the probability of choosing the outside alternative from all different menus. 
Since our objective is the identification of the consideration rule, we provide a natural restriction on the attention-index rule that guarantees invertibility of $\psi$. This restriction is satisfied by the models of interest in this paper (but not restricted to them). 
\begin{definition}[Totally monotone consideration]\label{def:totallymonotoneconsideration}
A consideration rule $m$ admitting an attention-index representation is \textit{totally monotone} if, we can write, for all $A\in\M$
\[
m_{A}(\emptyset)=\varphi\left(\eta(\emptyset),\sum_{C\subseteq A} \eta(C)\right),
\]
where $\varphi:[0,1]\times [0,1]\to[0,1]$ is a strictly monotone in each argument function.
\end{definition}
The probability of not considering any object, conditional on a given choice set, is assumed to be a monotonic function of the cumulative probability of paying attention to at least some alternative in the menu (according to $\eta$), and of the probability of not considering anything (unconditionally). Crucially, the dependence on the correspondence $g$ disappears in a totally monotone attention-index representation, with respect to the general attention-index representation. Note, this happens only for the case of  not considering anything (i.e., the consideration set is $\emptyset$).  
\par
Totally monotone attention-index rules are such that the random consideration is monotone as in \citet{cattaneo2017random}, namely $m_{A}(\emptyset)\leq m_{B}(\emptyset)$ if $B\subseteq A$, when $\varphi$ is strictly increasing in the first entry. However, in this case, they imply more. Since, the mapping $m_{(\cdot)}(\emptyset):\M\to [0,1]$ is a function of the cumulative probability of considering at least one item in any given menu (i.e., a function of $\sum_{C\subseteq A} \eta(C)$), the behavior of the probability of choosing the outside option will be restricted. For instance, for the case of $\rcg$ it will satisfy a form of marginal decreasing propensity of choice \citep{aguiar2017random}.
\par 
We highlight that total monotonicity is testable. Observe that the marginal probability of choosing the outside option, when a new set of alternatives is added to a menu, is weakly decreasing (e.g., $\Delta_{C}p(o,A)=p(o,A\cup C)-p(o,A)\leq 0$ and $\Delta_{D}(\Delta_{C}p(o,A)\leq0$)).\footnote{Note, however, that this restriction does not mean that $m$ satisfies the monotonicity property for other menus different from the empty set.} Alternatively, $\varphi$ can be strictly decreasing in the first entry, in which case it provides an antithetic behavior to that of RAM (yet testable). In this sense, this restriction on attention is neither weaker nor stronger than the monotonicity restriction in \citet{cattaneo2017random}. 
\par 
Strict monotonicity of $\varphi$ in each of its entries implies the invertibility of $\psi$.\footnote{This is a consequence of Mobius invertibility of the mapping $v(\cdot)=\sum_{C\subseteq \cdot}\eta (C)$ \citep{chateauneuf1989some}.} This sufficient condition for invertibility of the link function $\psi$ is mild. It holds in all the examples of interest in this work. Importantly, it is a testable restriction. Note that since the inverse of $\psi$ is known under the model of interest, we can compute a candidate $\eta$ from the data $P$. If the computed $\eta$ is not an element of the simplex $\Delta(2^X)$ then $\psi$ is not invertible. 
\par
The next lemma shows that the models we consider admit a totally monotone attention-index representation.
\begin{lemma}
Any $m\in\mathcal{M}^\rc$, $\rc\in\{\la,\mm,\rcg,\rum\}$, admits a totally monotone attention-index representation with 
\begin{itemize}
    \item $\varphi^{\la}(\eta_o,t)=\frac{\eta_o}{t}$; 
    \item $\varphi^{\rcg}(\eta_o,t)=1-t+\eta_{o}$;
    \item $\varphi^{\mm}(\eta_o,t)=\varphi^{\la}(\eta_o,t)$ and $\varphi^{\mm}(\eta_o,t)=\varphi^{\rcg}(\eta_o,t)$;
    \item $\varphi^{\rum}(\eta_o,t)=\varphi^{\rcg}(\eta_o,t)$.
\end{itemize}
\end{lemma}
The proof is omitted because of its simplicity for the cases of $\la$, $\rcg$, and $\rum$. For the case of $\mm$ the statement follows from \citet{brady2016menu} and \citet{aguiar2017random}.\footnote{Note that since $\mm$ and $\rum$ are special cases of $\la$ and $\rcg$ respectively, they share the same link function. However, empirically we will be able to differentiate among them because of the additional restrictions they pose on $\eta$.} 
\par 
The key assumption in this section has been that the default that is always present is always dominated. We formulate a sensitivity analysis when the default is not dominated in Appendix~\ref{appendix:sensitivitydefault}.
\subsection{Characterization of the \texorpdfstring{$\rhrc$}{Link-Behavioral}-model}
As a preliminary step for characterizing the $\rhrc$-model, we construct a candidate calibrated attention-index $\eta^\rc$ from the data $P$. Informally, this calibrated (revealed) attention-index is the result of inverting the link function $\varphi$ with respect to the probability of choosing the default alternative. The link function invertibility is a consequence of the monotonicity assumptions and the existence of a unique Mobius inverse of the cumulative attention-index $v(\cdot)=\sum_{C\subseteq \cdot }\eta(C)$ \citep{chateauneuf1989some}. We do this recursively. For a given $\varphi^{\rc}$, let $\varphi^{-1,\rc}_{1}$ and $\varphi^{-1,\rc}_{2}$ be the inverses of $\varphi^{\rc}$ with respect to the first and the second argument, respectively. Let $\abs{A}$ denote the cardinality of a finite set $A$. 

\begin{definition}[Calibrated attention-index] \label{def:calibrated eta} 
For given $P$, $\eta^{\rc}:2^X\to \Real$ is such that
(i) $\eta^{\rc}(\emptyset)=\varphi^{-1,\rc}_{1}(p(o,X),1)$, and (ii) for all $D\in 2^X\setminus X$
\[
\eta^{\rc}(D)=\sum_{B\subseteq D} (-1)^{\abs{D\setminus{B}}}\varphi^{-1,\rc}_{2}(\eta^{\rc}(\emptyset),p(o,D)).
\]
\end{definition}
The calibrated attention-index depends only on the dataset $P$ and the model $\rc$. If the calibrated attention-index of a set is negative, then $P$ could not have been generated by model $\rc$. This testable implication is analogous to the \citet{block_random_1960} inequalities. 
\par 
Now, we construct an object, $m^\rc$, that is a distribution over consideration sets if the model $\rc$ is correctly specified. 
\begin{definition}\label{def:calibrated m}
For a given $P$, let $m^{\rc}=(m_A^{\rc}(D))_{A\in\M,D\in 2^{A}}$, where $m^{\rc}_\A:2^{\A}\to\Real$ is such that for all $\A\in\M$ and $D\in 2^{A}$
\[
m_{A}^{\rc}(D)=\psi^\rc\left(\eta^\rc(D),\sum_{C\in g^\rc(D,A)}\eta^\rc(C)\right).
\]
\end{definition}
We can apply this generic formula for totally monotone attention-index rules to the specific models of interest.
\begin{example}\label{def:calibrated m ex}
\begin{itemize}
\item $m^{\la}_A(D)=\frac{\eta^{\la}(D)}{\sum_{C\subseteq A}\eta^{\la}(C)}$,
where $\eta^{\la}(D)=\sum_{B\subseteq D}(-1)^{\abs{D\setminus B}}\frac{p(o,X)}{p(o,B)}$;
\item $m^{\mm}_A(D)=\frac{\eta^{\mm}(D)}{\sum_{C\subseteq A}\eta^{\mm}(C)}$, where $\eta^{\mm}(D)=\prod_{a\in X\setminus D}\left(1-\gamma^{\mm}(a)\right)\prod_{b\in D}\gamma^{\mm}(b)$, and $\gamma^{\mm}:X\rightarrow \Real$ such that $\gamma^{\mm}(a)=1-\frac{p(o,A)}{p(o,A\setminus{\{a\}})}$ for some $A\in \M$ that contains $a$;
\item $m^{\rcg}_A(D)=\sum_{C:C\cap A=D}\eta^{\rcg}(C)$, where $\eta^{\rcg}(D)=\sum_{A\subseteq D:D\in \M}(-1)^{\abs{D\setminus{A}}} p(o,X\setminus{A})$;
\item $m^{\rum}_{A}(D)=\Char{A=D}$.
\end{itemize}
\end{example}
In general, $m^{\rc}$ may not be a distribution (some components may be negative or greater than $1$) since $m^{\rc}$ is calibrated from observed frequencies. Moreover, $m^{\la}$ or $m^{\mm}$ may not be well-defined if probabilities of choosing the outside option for some choice sets are zero. 

To be able to estimate $m^{\rc}$ from the data with probability approaching 1, we need the following definition that formalizes the above discussion.

\begin{definition} [Well-defined $m^{\rc}$]  $m^{\rc}$ is well-defined if $m_A^{\rc}\in \Delta(2^A)$ for all $A\in \M$.  
\end{definition}

We are ready to state our main result. 
\begin{theorem}\label{thm:calibration of m}
For every link function $\rc$, the following are equivalent: 
\begin{enumerate}
    \item $P$ is a $\rhrc$-rule; 
    \item $m^{\rc}$ is a well-defined and $P$ is a $\hrc$-rule described by $(m^{\rc},\pi)$.
\end{enumerate}
\end{theorem}

Theorem~\ref{thm:calibration of m} provides a full characterization of well-defined models with link functions. If $P$ is a $\rhrc$-rule, then $m^{\rc}$ has to be well-defined. Theorem~\ref{thm:calibration of m} implies that to test a given model one does not need to consider all possible distributions over considerations sets. It suffices to check the unique distribution that is calibrated from observed $P$ according to Definition~\ref{def:calibrated m}.
\par
Initially, we had to find two objects (the distribution over preferences $\pi$ and the distribution over consideration sets $m$) to make the data consistent with the model. Now we just need to find $\pi$. In other words, we simplified the testing problem. Unfortunately, the testing problem is still not tractable since the set of all possible distributions over preferences $\Delta(R(X))$ is big. To solve this problem we introduce another fictitious object.
\begin{definition}\label{def:calibrated P}
For given model $\rc$ and $P$, let $P_{\pi}^{\rc}=(p_{\pi}^{\rc}(a,\A))_{A\in\M,a\in X}$, where $p_{\pi}^{\rc}:X\times \M\to\Real$ is such that for all $\A\in\M$ and $a\in\A$
\begin{align*}
p^{\rc}_{\pi}(a,\A)=\dfrac{p(a,\A)-\sum_{C\subset\A}m^{\rc}_\A(C)p^{\rc}_{\pi}(a,C)}{m^{\rc}_{\A}(\A)}.
\end{align*}
\end{definition}
Note that when $P$ has been generated by a $\rhrc$-rule, $P_\pi^{\rc}$ corresponds to the underlying full-consideration random utility rule. In fact, we can write a $\rc\text{-}\hrc$ model equivalently as:
\[
p(a,A)=\sum_{D\subseteq A}m_A(D)p_\pi(a,D),
\]
where $p_\pi(a,A)=\sum_{\succ\in R(X)}\pi(\succ)\Char{a \succ b,\: \forall b\in A}$ is the underlying $\rum$ distribution over (nondefault) choices that is weighted by the random consideration rule $m$ to produce the observed behavior. When $P$ has been generated by this $\rc\text{-}\hrc$-rule, it follows that $p_\pi^{\rc}=p_\pi$. That is why we call $P_{\pi}^{\rc}$ \emph{the calibrated full consideration rule}.
\par
Similar to $m^\rc$, $P_{\pi}^\rc$ has interpretation when the $\rhrc$-rule is consistent with the data. The next theorem provides the last missing piece of our characterization before testing.
\begin{theorem}\label{thm:WRU characterization}
Suppose that for given model $\rc$ and stochastic choice rule $P$ (i) $m^{\rc}$ is well-defined, (ii) $m^{\rc}_{A}(A)>0$ for all $A\in\M$. Then the following are equivalent:
\begin{enumerate}
    \item $P$ is a $\rhrc$-rule;
    \item $P_{\pi}^{\rc}$ is a $\rum\text{-}\hrc$-rule.
\end{enumerate}
\end{theorem}
Note that both $m^{\rc}\geq 0$ and $P_{\pi}^{\rc}$ can be computed from $P$. Thus, Theorem~\ref{thm:WRU characterization} implies that to test a given model $\rc$ it is necessary and sufficient to test whether $m^{\rc}$ is well defined, and whether calibrated $P_{\pi}^{\rc}$ is a full consideration rule. 
Theorems~\ref{thm:calibration of m} and~\ref{thm:WRU characterization} provide a generalization of the characterization results in \citet{manzini2014stochastic, brady2016menu}, and \citet{aguiar2017random}. Moreover, they provides a unified result for all models that admit a (totally monotone) attention-index representation.\footnote{Note that Theorem~\ref{thm:WRU characterization} simplifies the testing problem because it avoids the problem of computing the distribution over choices for every consideration set in every menu. We only need to focus on computing the distribution over choices in each menu.}
\par
In practice, we do not observe $P$, but can consistently estimate it by the collection of sample frequencies $\hat{P}$. In Section~\ref{subsection: test} we discuss how to test the $\rhrc$-rule accounting for sampling variability in $\hat{P}$.

\subsection{Identification} \label{subsection:recoverability}
Assuming independence between the distribution of preferences and the random consideration set rule, we uniquely identify the consideration set rule from $P$ if it is a $\rhrc$-rule, for all models with totally monotone link functions. Moreover, if there is a positive mass of individuals that consider all alternatives in the choice set, the recoverability of preferences is as good as in the case of full consideration. 
 
\begin{theorem}[Identification]\label{thm:identification}
Suppose that for a given model $\rc$ (i) $P$ is a $\rhrc$-rule and (ii) $m^{\rc}_{A}(A)>0$ for all $A\in\M$. If $P$ is described by $(m,\pi)$ and $(m',\pi')$, then $m=m'$ and $p_{\pi}=p_{\pi'}$.
\end{theorem}

We underline that we achieve a unique decomposition of the dataset $P$ into its attention and preference components. 
Identification of preferences and consideration rules is not a trivial task. Even for simple datasets where stochastic behavior arises from only one channel (for example limited consideration), models that only allow for stochastic behavior because of preference heterogeneity (e.g. $\rumo$), or because of random consideration/attention without additional assumptions (e.g. RAM) may fail to identify underlying preferences even when they perfectly describe observed choices. Our framework shows that the recoverability of preferences is as good as the $\rumo$ benchmark in stark contrast with RAM, where nothing can be learned about preferences for regular random choice.

\section{Frames and Testing\label{section:frames}}
\subsection{Frames}
In this section, we introduce another source of variation in the data--attention frames. This additional source of variation will allow us to differentiate between behavior consistent with $\rumo$ and behavior consistent with $\rhrc$. Following \citet{salant2008f}, we define the extended choice set as a pair of a choice set $A\in \M$ and an attention frame $f\in F$. In this extended environment, for a given frame $f\in F$, we can define a probabilistic choice rule $p_f$. Similarly, a complete stochastic choice rule with frame $f\in F$ is $P_f=(p_f(a,A))_{A\in\mathcal{A},a\in A\cup\{o\}}$. 
\par
The elements of $F$ contain descriptions of physical items that only vary in presentation but not in the information they contain. Our interpretation of attention frames is the same as \citet{bhattacharya2021frame}. These descriptions are available to DMs and should not affect their preferences, but may influence their attention. 
\par 
\begin{example}[Stochastic choice rule with frames]\label{ex: frame}
Let $X=\{a,b,c\}$ and consider two frames: (i) $f$ that describes $a=1$ token, $b=2$ tokens and $c=3$ tokens; (ii) and $f'$ that describes $a=3-2$ tokens, $b=10-8$ tokens, and $c=2+1$ tokens. An example of an (incomplete) stochastic choice rule with frame $f$ is $P_f=(p_f(a,\{a,b\})=0,p_f(b,\{a,b\})=1,p_f(b,\{b,c\})=0,p_f(c,\{b,c\})=1).$ An example of an (incomplete) stochastic choice rule with frame $f'$ is $P_{f'}=(p_{f'}(a,\{a,b\})=1/2,p_{f'}(b,\{a,b\})=1/2,p_{f'}(b,\{b,c\})=1/2,p_{f'}(c,\{b,c\})=1/2).$
\end{example}
In Example~\ref{ex: frame}, $f$ and $f'$ present the same information about the same alternatives in two distinct ways. In particular, different frames correspond to different numbers of arithmetic operations used to describe the value of the option. The same value is expressed in each frame with a different sequence of arithmetic operations.
\par 
The $\hrc$-rule with frame $f$ can be defined analogously to the definition of the $\hrc$-rule in Section~\ref{section:model}. A complete stochastic choice rule with frame $f$, $P_f$, is a $B$-rule if there exists a frame dependent distribution over random consideration sets, $m_{f}$ and a frame dependent distribution over strict linear orders $\pi_{f}$ such that
\[
p_f(a,A)=\sum_{\succ\in R(X)}\pi_f(\succ)\sum_{D\subseteq A}m_{f,A}(D)\Char{a\succ b,\forall b\in D}.
\]
\par 
We take the stand that preferences should not depend on the way the information is presented. Only attention may change due to frames. We formalize these ideas by the assumption of consequentialism, which is a common implicit assumption in the RUM framework of \citet{mcfadden1990stochastic}. 
\begin{definition}[Consequentialism]\label{def: consequentialism}
 A collection of $\hrc$-rules $(P_f)_{f\in F}$ described by  $((m_f,\pi_f))_{f\in F}$  is said to satisfy consequentialism if there exists $\pi$ such that $\pi_f=\pi$ for all $f\in F$.
\end{definition}
Consequentialism means that an attention frame does not alter the payoffs that a DM obtains from choosing a given alternative. 
In the same way that the standard rational choice framework imposes frame independence \citep{salant2008f}, the classical RUM imposes consequentialism \citep{mcfadden1990stochastic}.
\par
At this point, it is useful to formally define the Random Utility Model ($\rumo$) over the whole choice set $X\cup \{o\}$. $\rumo$ treats the default alternative as just another item with no special status. That is, it is not assumed to be a dominated alternative.
Let $R(X\cup\{o\})$ be a set of linear orders over the extended choice set $X\cup\{o\}$.
\begin{definition}[Random Utility Model, $\rumo$]\label{definition:rumo}
A collection of complete stochastic choice rules with frame  $(P_f)_{f\in F}$ is consistent with random utility if there exists $\pi_o\in \Delta(R(X\cup\{o\}))$ such that
\[
p_f(a,A)=\sum_{\succ\in R(X\cup\{o\})}\pi_o(\succ)\Char{a\succ b,\forall b\in A},
\]
for all $a\in A$, $A\in\M$, and all $f\in F$.
\end{definition} 
Note that $\rumo$ satisfies consequentialism and $\pi_o$ does not depend on the frame. The $\rhrc$-rule extends the RUM framework to allow for frame dependence only through the random consideration rule.
\par
By assuming that preferences satisfy consequentialism, we are essentially making two assumptions: (i) we assume that the distribution of preferences does not depend on how the alternatives are presented; (ii) consideration in turn can depend on payoff irrelevant information. These assumptions are important, since our experiment has three treatments where a subject must perform a different number of arithmetic operations to evaluate the monetary value of a prize. 
\par 
These assumptions are auxiliary and are not necessary for either $\rumo$ or $\rhrc$-rules. However, these assumptions are a reasonable baseline to compare $\rumo$ and random consideration models. The reason these assumptions are not immediately implied is because $\rumo$ is often viewed as descriptive. Taking this one step further, one could assume that the distribution of preferences varies with the difficulty of evaluating a task in our experiment. This would violate our assumption of consequentialism. While frame-dependent preference could be a reasonable assumption in some settings, it would limit the ability to predict counterfactual choices. For example, one would need to identify  preferences in each setting to get a prediction on choices.
\par
Specifically, frame variation allows us to differentiate between $\rhrc$-rules and $\rumo$ as follows. Without it $\la$ is not nested by nor nests $\rumo$. For example, $\la$ allows for attraction effect which violates regularity and therefore is inconsistent with $\rumo$. Also, their intersection is nonempty because $\mm$ is both consistent with $\rumo$ and $\la$ \citep{manzini2014stochastic,brady2016menu}. Moreover, for a fixed frame $\rcg$ is nested in $\rumo$ and nests $\mm$, therefore its intersection with $\la$ is nonempty \citep{aguiar2017random} (see Figure~\ref{fig:relations}).\footnote{\citet{suleymanov} shows that $\mm$ is the intersection of $\la$ and $\rcg$.} However, with frame variation this is no longer true: $\rumo$ and $\rcg$ intersect, but are not nested within each other.\footnote{All these relationships are preserved when allowing for heterogeneous preferences under the independence assumption of preferences and attention. The reason is that the outside probability does not depend on the distribution of preferences.  For more details see Appendix~\ref{app:section_models}.} 

\begin{figure}[t]
		\centering
		\begin{tikzpicture}	[scale=1.1]
		\draw [line width=1.5pt ](2,2) -- (2,6) -- (6,6)--(6,2) --(2,2);
		\draw [dotted, line width=1.5pt ](3,2.75) -- (3,5.5)-- (5.25,5.5)--(5.25,2.75)--(3,2.75);
		\draw (4.15,2) node[above] {$\rumo$};
		\draw (4.15,3.35) node[above] {$\mm$};
		\draw (4.15,5) node[above] {$\rcg$};
\draw (1.25,3.35) node[above] {$\la$};
		\draw [dashed,line width=1.5pt ](.5,2.75) -- (.5,4.5)-- (5.25,4.5)--(5.25,2.75)--(.5,2.75);
		\end{tikzpicture}	
	\caption{\textsc{Relation among Consideration Set Rules for a fixed frame: $\rumo$, $\la$, $\mm$, and $\rcg$} } \label{fig:relations}
\end{figure}

\subsection{Testing Procedure} \label{subsection: test}
Theorem~\ref{thm:WRU characterization} allows us to test whether for a given frame $f$, a given stochastic choice rule $P_f$ is a $\rhrc$-rule: it is necessary and sufficient to test whether $m_f^{\rc}$ is well-defined (satisfies a set of linear inequalities) and $P_{f,\pi}^{\rc}$ is consistent with the full consideration model. Note that the full consideration model is equivalent to the random utility model without outside option. Testing for $\rumo$ is a well-understood problem and amounts to solving a quadratic optimization with cone constraints (see \citealp{mcfadden1990stochastic} and \citealp{kitamura2018nonparametric}). The approach proposed by \citet{kitamura2018nonparametric} allows us to test these conditions while accounting for sampling variability induced by using $\hat{P}_f$ instead of unknown $P_f$. 
\par
We, however, need to slightly modify the testing procedure in \citet{kitamura2018nonparametric} to take into account the frame variation. First we describe the testing procedure for the fixed frame and then extended it to environments with frame variation.
\par
To introduce the testing procedure, we need to define several objects. Note that, for a fixed frame $f$, the calibrated full consideration rule, $P_{f,\pi}^{\rc}$, is a vector of length $d_p=\sum_{k=1}^{\abs{X}}k{\abs{X} \choose k}$.\footnote{${n \choose k}=\dfrac{n!}{k!(n-k)!}$ and $n!=1\cdot2\cdot\dots\cdot n$.} The $k$-th element of $P_{f,\pi}^{\rc}$ corresponds to some pair $(a,\A)$ such that $a\in\A$. 

Let $B_1$ be the matrix of the size $d_p\times \abs{X}!$ such that $(k,l)$ element of it is equal to
\[
B_{1,k,l}=\Char{a\in\A}\Char{a\rr_l c,\:\forall\: c\in\A},
\]
where $k$ corresponds to a pair $(a,A)$ such that $a\in\A$, and $\rr_l$ is $l$-th linear order on $X$. We define $G_1$ as the matrix of size $(d_p+d_m)\times d_1$, where $d_1=\abs{X}!+d_m$ and $d_m=\sum_{A\subseteq X}2^{\abs{A}}$ is the dimension of $m_f^{\rc}$, such that
\[
G_1=\left[\begin{array}{cc}
     B_1 &0_{d_p\times d_m}\\
     0_{d_m \times \abs{X}!}&I_{d_m} 
\end{array}
\right],
\]
where $0_{d_p\times d_m}$ denotes the zero matrix of size $d_p\times d_m$, and $I_{d_m}$ denotes the identity matrix of size $d_m\times d_m$.
The next result establishes an equivalent characterization of the $\rhrc$-rule via $m_f^{\rc}$ and $P_{f,\pi}^{\rc}$. Let $\Real^{d_1}_{+}$ denote component wise nonnegative elements of the $d_1$-dimensional Euclidean space $\Real^{d_1}$.
\begin{theorem}\label{thm:WRU characterization2}
For a fixed frame $f$ the following are equivalent:
\begin{enumerate}
    \item $P_{f,\pi}^{\rc}$ is a $\rum\text{-}\hrc$-rule and $m_f^{\rc}$ is well-defined;
    \item $\inf_{v\in\Real^{d_1}_{+}}\norm{g_f^{\rc}-G_fv}=0$, where $g_f^{\rc}=(P_{f,\pi}^{\rc\prime},m_f^{\rc\prime})\tr$.
\end{enumerate}
\end{theorem}
\begin{proof}
See \citet{mcfadden1990stochastic} and \citet{kitamura2018nonparametric}.
\end{proof}
Theorem~\ref{thm:WRU characterization2} implies that we can test the null hypothesis that $\inf_{v\in\Real^{d_1}_{+}}\norm{g_f^{\rc}-G_1v}=0$. Fortunately, this testing problem can be directly cast to the testing problem in \citet{kitamura2018nonparametric}. 
\par
To take into account the frame variation and consequentialism, we need to modify the matrix $G_1$. Let $d_f=\abs{F}$ and $B$ be a matrix of size $d_f\cdot d_p\times\abs{X}!$ that consists of $d_f$ matrices $B_1$ stacked together. That is, $B=\left(B_1\tr \:B_1\tr\:\dots\:B_1\tr\right)\tr$. Let $d=\abs{X}!+d_f\cdot d_m$. Similarly to $G_1$ define,
\[
G=\left[\begin{array}{cc}
     B &0_{d_f\cdot d_p\times d_f\cdot d_m}\\
     0_{d_f\cdot d_m \times \abs{X}!}&I_{d_f\cdot d_m} 
\end{array}
\right].
\]
Note that when $d_f=1$ (i.e., the frame is fixed), then $G=G_1$.
\begin{cor}\label{cor:WRU characterization2}
The following are equivalent:
\begin{enumerate}
    \item $(P_f)_{f\in F}$ satisfies consequentialism; $P_{f,\pi}^{\rc}$ is a $\rum\text{-}\hrc$-rule and $m_f^{\rc}$ is well-defined for all $f\in F$.
    \item $\inf_{v\in\Real^{d}_{+}}\norm{g^{\rc}-Gv}=0$, where $g^{\rc}=\left((P^\rc_{f,\pi})\tr_{f\in F}, m_f^{\rc})\tr_{f\in F}\right)\tr$.
\end{enumerate}
\end{cor}

\section{The Experiment}\label{section:experiment}

Our testing approach does not have requirements in terms of repeated individual choice data. Exploiting this feature, our experiment was designed to study the performance of different theories of random consideration sets with few observations per individual. In particular, we conducted the experiment in Amazon MTurk for a large cross-section with at most two (disjoint) choice sets per individual (see Section~\ref{subsection:experimental_design}). The large sample size of the dataset generated by our experiment is fundamental for ensuring high statistical power using the tools in \citet{kitamura2018nonparametric}. 

All sessions were run between August 25, 2018 and September 17, 2018 on the MTurk platform with surveys designed in Qualtrics.\footnote{By clicking the link on the MTurk page, subjects were randomly directed to one of the treatments implemented by Qualtrics. After completing their task, subjects were also asked to complete a short survey regarding their demographic information. Subjects were not allowed to participate in the experiment more than once. Only subjects living in the USA were recruited.} 
We surveyed $2135$ individuals. They were paid on average $\$1.09$ as a result of $\$0.25$ for participation fee and the outcome of a randomly selected task that pays a minimum of $\$0$ and a maximum of $\$2$. All payments were made in USD. The average duration of the session was $251.68$ seconds (slightly over 4 minutes).\footnote{The average duration of each task is about $23$ seconds, and the duration is significantly correlated with the length of the choice set and the frame.} This means that our average payment per hour is roughly $\$15$.

The payment in our experiment is comparable to other well-known experiments conducted in MTurk. To name a few, \citet{horton2011online} studied behavior in MTurk using games with the payment range between $\$0.40$ and $\$1.60$. They find that behavior in MTurk is consistent with behavior in the lab, where the stakes of games are ten times bigger. They also estimate the median minimum wage in MTurk as $\$0.14$ per hour.\footnote{The minimum wage here refers to the reservation wage the MTurk subjects have for performing a given task.} \citet{dean2014preference} conducted experiments of decision-making. The average payment for completing 15-min long tasks was between $\$1.35$ and $\$1.55$ including the show-up fee of $\$0.25$. \citet{kim2016discounting} conducted an experiment in MTurk for several weeks with one 10-min task each week. The average earnings from each week's task were below $\$1.00$. \citet{rand2012spontaneous} also conducted a public good game with MTurkers and the payment range was between $\$0.90$ and $\$1.50$ including the show-up fee of $\$0.50$.

\subsection{Experimental Design} \label{subsection:experimental_design}
In our experiment, we have two independent sources of exogenous variation: full variation in choice sets, and variation in frames. Recall that full variation in choice sets means that all possible choice sets are observed, and variation in frames means that we vary observable information without affecting the rational assessment of the alternatives. These two sources of variation allow us to test consideration-mediated choice theories in a large cross-section of heterogeneous individuals, and differentiate these theories from $\rumo$. The latter is possible since consideration is allowed to change with frames, but preferences must remain stable due to consequentialism. We vary the frame through changes in the cost of consideration.

\subsection*{Choice set design}
To induce preference heterogeneity, we consider lottery alternatives with different expected values and variances. Table~\ref{lotteries} shows the alternatives and implied preference rankings if DMs are expected utility maximizers with CRRA Bernoulli utility function. The outside option is dominated for moderate levels of risk aversion (e.g. \citealp{holt2002risk}).\footnote{Recall that the assumption that the default is dominated is a testable assumption in our framework. In addition, without cost of consideration treatments, the outside alternative is easier to understand than the rest because of its simplicity. Hence, it works as a consideration-reference point in the sense of \citet{suleymanov}.}

\begin{table}[ht]	\caption[Lotteries]{\textsc{Lotteries measured in tokens, expected values, and variance}}\label{lotteries}  \centering
\resizebox{0.95\textwidth}{!}{\begin{tabular}{@{\extracolsep{1pt}}|c l| c | c | p{9mm}p{9mm}p{9mm}p{9mm}p{9mm}p{9mm} |@{}}
\cline{1-10}
\multicolumn{2}{|c|}{} &  &  &\multicolumn{6}{c|}{}  \\	[.25ex] 
\multicolumn{2}{|c|}{\textsc{Lottery}} & \textsc{Expectation} & \textsc{Variance} &\multicolumn{6}{c|}{\textsc{Preference Rank $u(x)=\frac{x^{1-\sigma}}{1-\sigma}$ with $\sigma$}}  \\ [2ex]
 &&&&-2&0&0.25&0.30&0.50&0.75\\ [.1ex] 	\cline{1-10} 
 &&&&&&&&&\\ [.1ex]
(1)&$\frac{1}{2} 50 + \frac{1}{2} 0$ & 25.000 & 625.00   &1&1&2&5&5&6 \\[1ex]
(2)&$\frac{1}{2} 30 + \frac{1}{2} 10$ & 20.000 & 100.00 &5&5&5&2&1&1\\[1ex]
(3)&$\frac{1}{4} 50 + \frac{1}{4} 30 + \frac{1}{4} 10 + \frac{1}{4} 0$ & 22.500 & 368.75 &3&3&4&4&3&4\\[1ex]
(4)&$\frac{1}{4} 50 + \frac{1}{5} 48 + \frac{3}{20} 14 + \frac{2}{5} 0$ & 24.125 & 511.73 &2&2&1&3&4&5\\[1ex]
(5)&$\frac{1}{5} 48 + \frac{1}{4} 30 + \frac{3}{20} 14 + \frac{1}{4} 10+ \frac{3}{20} 0$ & 21.625 & 251.11&4&4&3&1&2&3 \\[1ex]
(o) & 12 with probability 1  & 12.000 & 0.00 &6&6&6&6&6&2\\ \cline{1-10}\end{tabular}}
\end{table}

Let $X=\{l_1,l_2, l_3, l_4, l_5\}$ be the set of all nondefault alternatives, and let $o$ be the default/outside option. All menus $A\in\M$ are observed in the sample. The outside option is always present and is shown first, while the order of other alternatives is randomized. Menus can be thought as different treatments. 

Our primitive to test $\rhrc$ is $\hat{P}=\left(\hat{p}(a,A)\right)_{a\in A\cup\{o\},A\in\M}$, therefore we proceeded with stratified sampling, setting the minimal number of observations per choice set to be proportional to its cardinality, i.e. $n_A=\lambda (\abs{A}+1)$ with $\lambda\geq30$. This design requires a minimum of $\sum_{A\in \M}\abs{A}=3330$ tasks. 

For each menu, the DM faced three consideration frames or cost-treatments: High (H), Medium (M), and Low (L). These frames/cost-treatments were induced by introducing a $k$-length two digit addition/subtraction to compute each prize in the lottery. The length $k$ was set equal to $5$, $3$, and $1$, for the high, the medium, and the low cost, respectively. Since in our experiment attention frames only change the complexity of the description of lotteries, we assume that preferences of DMs do not depend on the way alternatives are described, thus are consistent with consequentialism. 
\par
The numbers for the cognitive task were randomly generated. Examples can be seen in Figure~\ref{dif_costs}. The default alternative $o$ was presented as is, and there was no need to solve an arithmetic problem to understand it across the different levels of cost.

\begin{figure}
	\begin{subfigure}[b]{0.48\textwidth}
		\centering
		\includegraphics[width=0.98\textwidth]{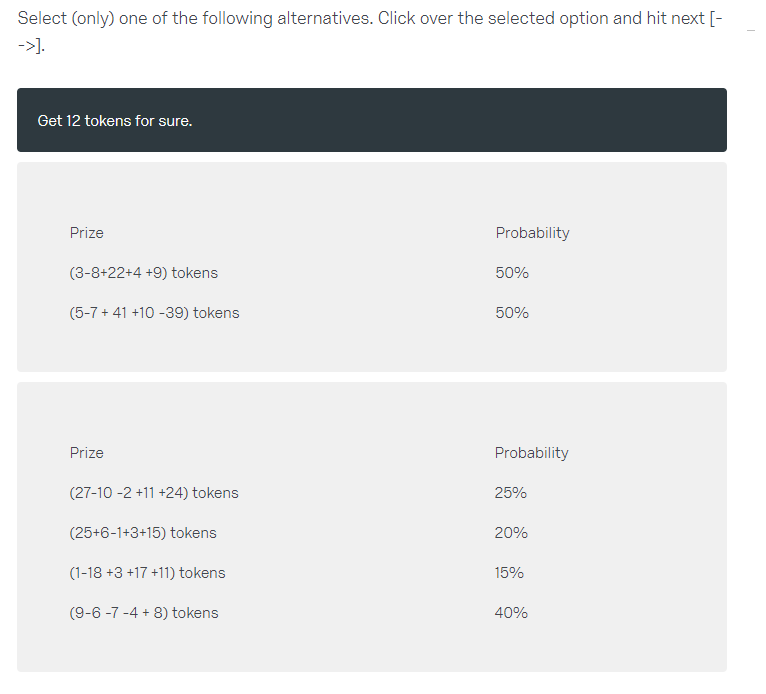}
		\caption{High Cost}
	\end{subfigure}
	\begin{subfigure}[b]{0.48\textwidth}
		\centering
		\includegraphics[width=0.98\textwidth]{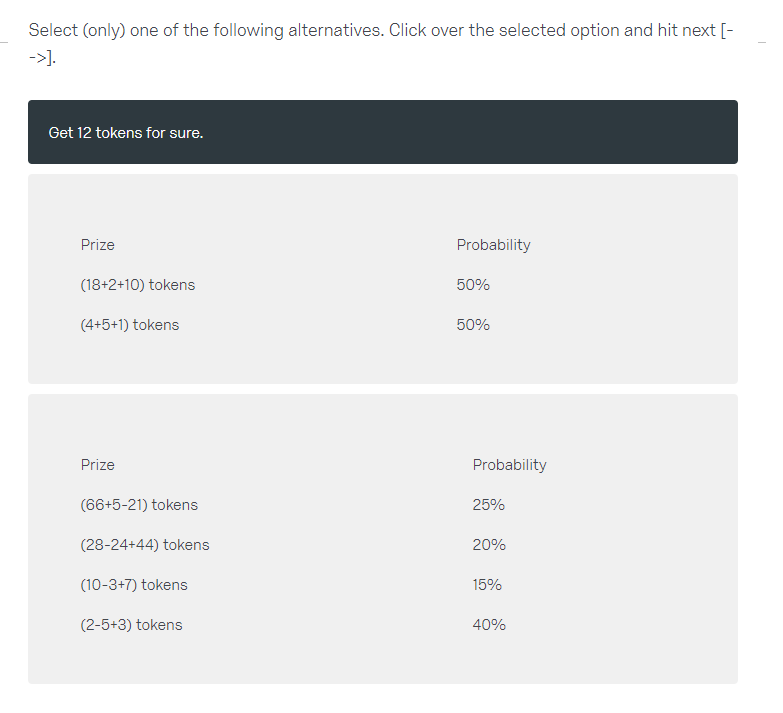}
		\caption{Medium Cost}
	\end{subfigure}
	\centering
	\begin{subfigure}[b]{0.75\textwidth}
		\centering
		\includegraphics[width=0.62\textwidth]{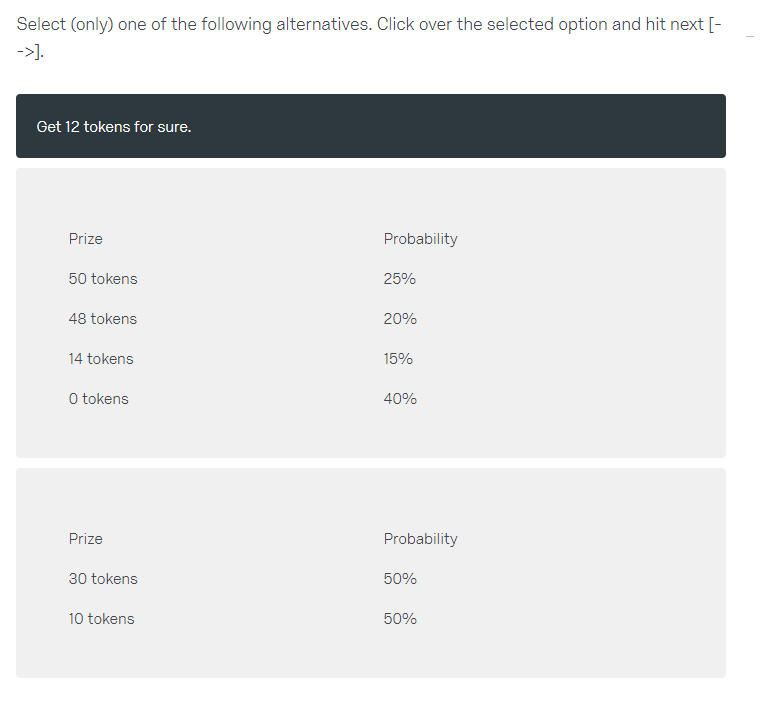}
		\caption{Low Cost} 
	\end{subfigure}
	\caption{\textsc{Consideration Cost Treatments } Different induced costs for Choice set $\{o, l_2, l_4\}$ } \label{dif_costs}
\end{figure}

To prevent possible learning that could attenuate consideration costs, subjects were faced with disjoint choice sets. That is, subjects were either presented with the full choice set and the outside option ($X\cup\{o\}$); or a partition of $X$ (presented at random order), i.e. $A_j\cup\{o\},A_k\cup\{o\}$ with $A_j\cup A_k=X$ and $A_j\cap A_k=\emptyset$. Our experimental design is summarized in Figure~\ref{figure:experiment_design}.
\paragraph{The default alternative} Our design allows us to use $o$ as the opportunity cost of incurring in the cost of consideration and understanding the other lotteries in the choice set. We use a degenerate lottery as the default due to its simplicity. In this sense, we believe the alternative $o$ in our design has effectively zero cost of consideration. Recall, that for any choice set/frame the outside option is always present and shown first. Moreover, it is pre-selected as the default alternative. If the subject decides to skip the task, she is informed that $o$ will be chosen for her. 

\begin{figure}[t]
\centering
\begin{tikzpicture}	[scale=1.05]
\draw [line width=1.25pt](5,10) ellipse (3cm and 0.5cm);
\draw (5,9.75) node[above] {\footnotesize{\textsc{DM - MTurk}}};
\draw [thick,->,>=stealth, line width=2pt] (5,9.5)--(5,8.5);
\draw [dotted, line width=1.5pt ] (3,8.5) -- (3,7.75) -- (7,7.75)--(7,8.5) --(3,8.5);
\draw (5,7.85) node[above] {\footnotesize{\textsc{Stratification}}};
\draw (5.8,7.1) node[right] {\footnotesize{$prob(A_j)=q_j\propto \left(\abs{A_j}+1\right)$}};
\draw (-2.2,7.1) node[right] {\footnotesize{$A_2\cup A_j=X$ and $A_2\cap A_j=\emptyset$}};
\draw [thick,->,>=stealth, line width=2pt,gray] (4,7.5)--(1.5,6.5);
\draw [thick,->,>=stealth, line width=2pt] (5.75,7.5)--(5.75,6.5);
\draw[fill=blue!45] (-0.75,6.25)--(-0.75,5.5)--(.35,5.5)--(.35,6.25)--(-0.75,6.25);
\draw[fill=red!20] (-4.5,5.15)--(-4.5,4.4)--(-1,4.4)--(-1,5.15)--(-4.5,5.15);
\draw (-2.75,4.55) node[above] {\footnotesize{(H) $a+b+c+d+e$}};
\draw (-2.85,5.6) node[above] {\footnotesize{\textsc{Prizes per Treatment}}};
\draw (-5,4) node[above, rotate=90] {\footnotesize{\textsc{Consideration cost treatments}}};
\draw (10.5,4) node[above, rotate=270] {\footnotesize{\textsc{Same menu/pay at random}}};
\draw[fill=yellow!20] (-4.5,3.85)--(-4.5,3.15)--(-1,3.15)--(-1,3.85)--(-4.5,3.85);
\draw[fill=green!20] (-4.5,2.6)--(-4.5,1.85)--(-1,1.85)--(-1,2.6)--(-4.5,2.6);		
\draw (-.2,5.6) node[above] {\scriptsize{$A_1\cup \{o\}$}};
\draw (-.2,4.5) node[above] {\scriptsize{$\{o,l_1\}$}};
\draw (-.2,3.25) node[above] {\scriptsize{$\{o,l_1\}$}};
\draw (-.2,2) node[above] {\scriptsize{$\{o,l_1\}$}};
\draw[fill=blue!35] (0.75,6.25)--(0.75,5.5)--(1.8,5.5)--(1.8,6.25)--(0.75,6.25);
\draw (1.25,5.6) node[above] {\scriptsize{$A_2\cup\{o\}$}};		
\draw (1.25,4.5) node[above] {\scriptsize{$\{o,l_1,l_2\}$}};
\draw (1.25,3.25) node[above] {\scriptsize{$\{o,l_1,l_2\}$}};
\draw (1.25,2) node[above] {\scriptsize{$\{o,l_1,l_2\}$}};
\draw[fill=blue!25] (2.2,6.25)--(2.2,5.5)--(3.3,5.5)--(3.3,6.25)--(2.2,6.25);
\draw (2.75,5.6) node[above] {\scriptsize{$A_3\cup\{o\}$}};
\draw (2.7,4.5) node[above] {\scriptsize{$\{o,l_1,l_2,l_3\}$}};
\draw (2.7,3.25) node[above] {\scriptsize{$\{o,l_1,l_2,l_3\}$}};
\draw (2.7,2) node[above] {\scriptsize{$\{o,l_1,l_2,l_3\}$}};
\draw [dotted, line width=1.5pt] (3.3,5.9)--(5,5.9);
\draw [dotted, line width=1.5pt] (6.25,5.9)--(8.25,5.9);
\draw[fill=blue!15] (5.1,6.25)--(5.1,5.5)--(6.2,5.5)--(6.2,6.25)--(5.1,6.25);
\draw (5.65,5.6) node[above] {\scriptsize{$A_{j}\cup\{o\}$}};
\draw (5.65,4.5) node[above] {\scriptsize{$\{o,l_3,l_4,l_5\}$}};
\draw (5.65,3.25) node[above] {\scriptsize{$\{o,l_3,l_4,l_5\}$}};
\draw (5.65,2) node[above] {\scriptsize{$\{o,l_3,l_4,l_5\}$}};
\draw[fill=blue!10] (8.3,6.25)--(8.3,5.5)--(9.5,5.5)--(9.5,6.25)--(8.3,6.25);
\draw (8.9,5.6) node[above] {\scriptsize{$A_{31}\cup\{o\}$}};
\draw (8.9,4.5) node[above] {\scriptsize{$\{o,l_1,l_2,l_3,l_4,l_5\}$}};
\draw (8.9,3.25) node[above] {\scriptsize{$\{o,l_1,l_2,l_3,l_4,l_5\}$}};
\draw (8.9,2) node[above] {\scriptsize{$\{o,l_1,l_2,l_3,l_4,l_5\}$}};
\draw (-2.75,3.2) node[above] {\footnotesize{(M) $f+g+h$}};
\draw (-2.75,1.9) node[above] {\footnotesize{(L) $k$}};
\end{tikzpicture}	
	\caption{\textsc{Experimental Design:} DM $i$ draws with probability $p(A_j)$ menu $A_j$ with $\abs{A_j}\in\{3,4,5\}$. In the picture, $A_j=\{l_3,l_4,l_5\}$. Therefore she is asked to choose from menus $A_j\cup\{o\}$ and $A_2\cup\{o\}$, since $A_2\cup A_j=X$ and $A_2\cap A_j=\emptyset$. She is faced with one of these menus first (randomly selected) and asked to choose when the cost is H, M and L. Then, she faced the other menu for the three cost treatments.   } \label{figure:experiment_design}
\end{figure}

\subsection{Sample}\label{subsection:sample}
The sample consists of 2135 individuals that selected alternatives from one or two choice sets for all costs of attention, as shown in Figure~\ref{figure:experiment_design}, for a total of 12297 observations. The number of observations per alternative/choice set are shown in Table~\ref{observations}. Based on these observations the primitive for our analysis is the collection of observed frequencies $(\hat p(a,A))_{a\in A,\,A\in\M}$. We compute these frequencies for all costs.
Unless otherwise stated $\hat{p}(a,A)$ refers to observed frequency in the data pooled across attention costs.  
\begin{figure}
		\centering
		\includegraphics[width=1\textwidth]{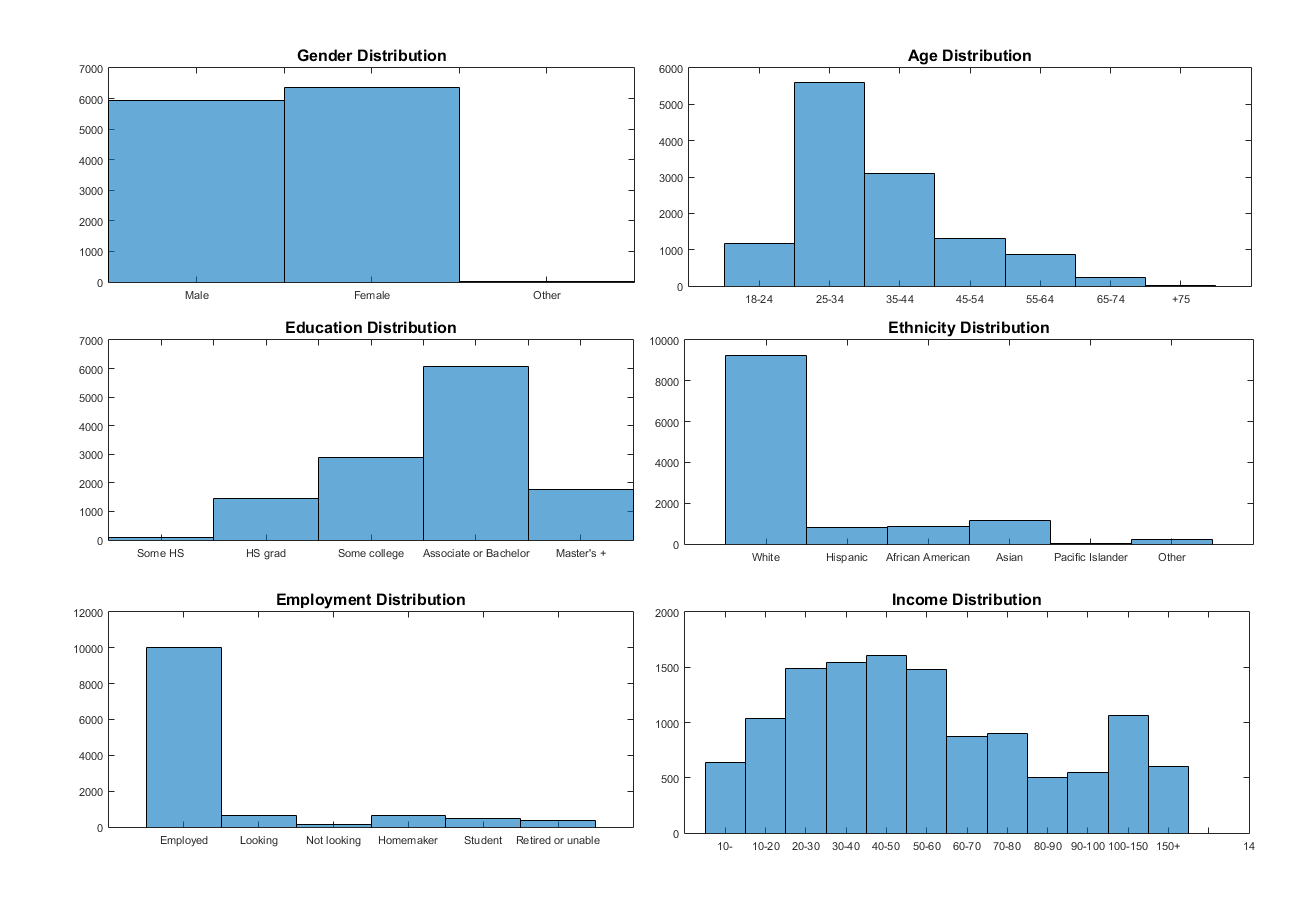}
	\caption{\textsc{Distribution of demographics in sample} } \label{dist_demographics}
\end{figure}

Figure~\ref{dist_demographics} summarizes the distribution of gender, age, education, ethnicity, labor, and income in our sample. Our subjects are a diverse sample of US individuals. By design, demographics are balanced across consideration cost treatments and choice sets (that can also be thought of as treatments). 

\subsection{Descriptive Analysis: Evidence for Costly Consideration and Frame Effects} \label{subsection:comparative_statistics}
In this section, we describe the behavior of individuals in our sample and present suggestive evidence that our cost treatments effectively induce costly consideration or frame effects. In particular, we observe that the consideration cost treatments/frames: (i)  affect the choice frequency of the outside option; (ii) have a heterogeneous effect on the choice frequencies of all other alternatives; (iii) and affect the patterns of choice with respect to the size of the menu. Moreover, all these effects depend monotonically on the level of difficulty of choice induced by each treatment. 
\par
Under the null hypothesis of full consideration and consequentialism, the observed frequency of choice of the outside option should remain the same across frames. The reason is that the choice menu remains the same across cost treatments, and payment is at random. However, the outside option is chosen more often as the cost increases (see Figure~\ref{fig:prob_outside}). This is evidence against full consideration and in favor of frame-dependent choice.

\begin{figure}	
\centering
\includegraphics[width=.6\textwidth]{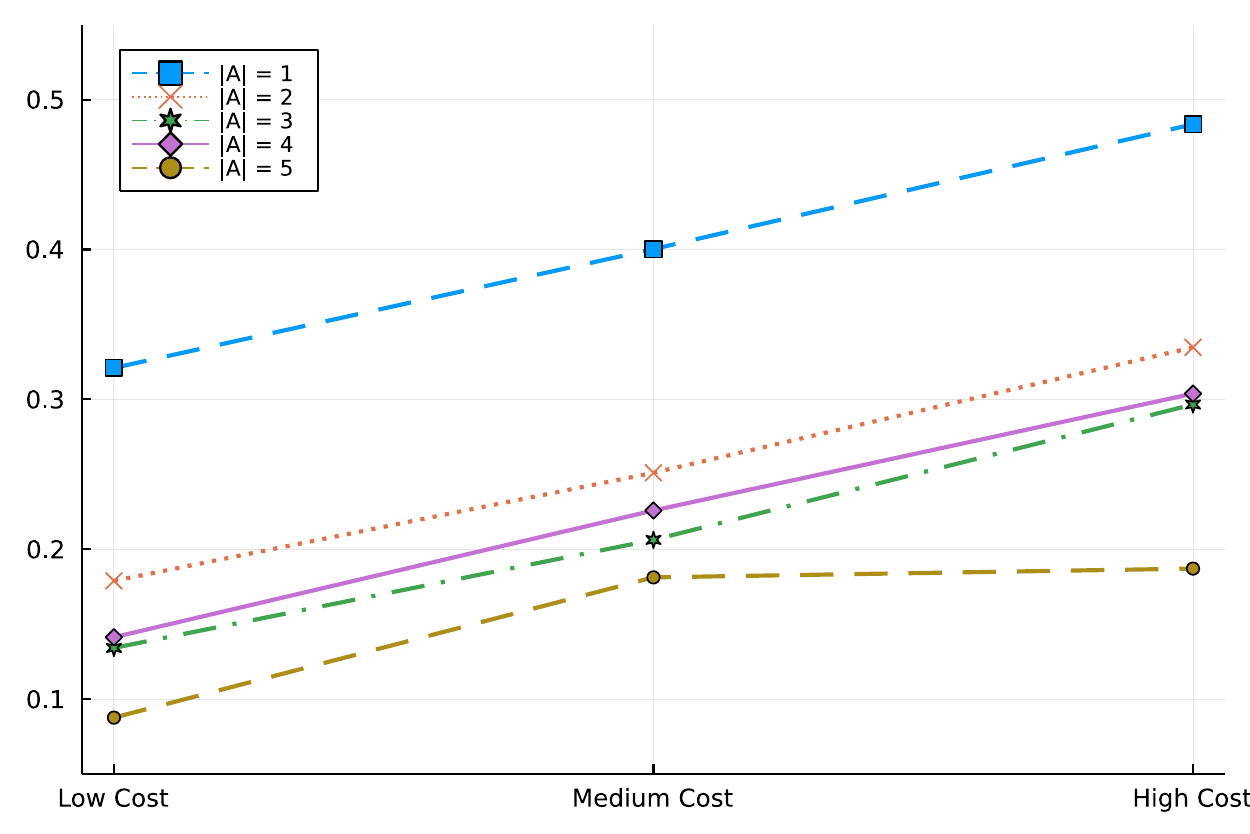}
\caption{\textsc{Frequency of the Outside Option Conditional on the Menu size and Frame.}}\label{fig:prob_outside}
\end{figure}
	
We remind the reader here that low, medium, and high cost corresponds to $1$, $3$, and $5$ arithmetic operations required to understand the monetary prize of each lottery, respectively. The monotone relation shows that our treatments were effective, and that the frequency of choice of the default is in fact ordered in the way it was expected. 
\par
Figure~\ref{fig:choices_cost} shows the effect of the frames/cost-treatments on the choice frequencies. The harder it is to understand the lotteries,\footnote{Here simplicity comes in the form of how easy (number of arithmetic operations) it is to compute expectation, variance, and expected utility of the lottery in terms of the number of prizes and whether the probabilities are uniform on the support of the lottery or not.} the more likely subjects opt to not consider them and instead choose the outside option. These results support the effectiveness of the induced treatments.
\begin{figure}
		\centering
		\includegraphics[width=.9\textwidth]{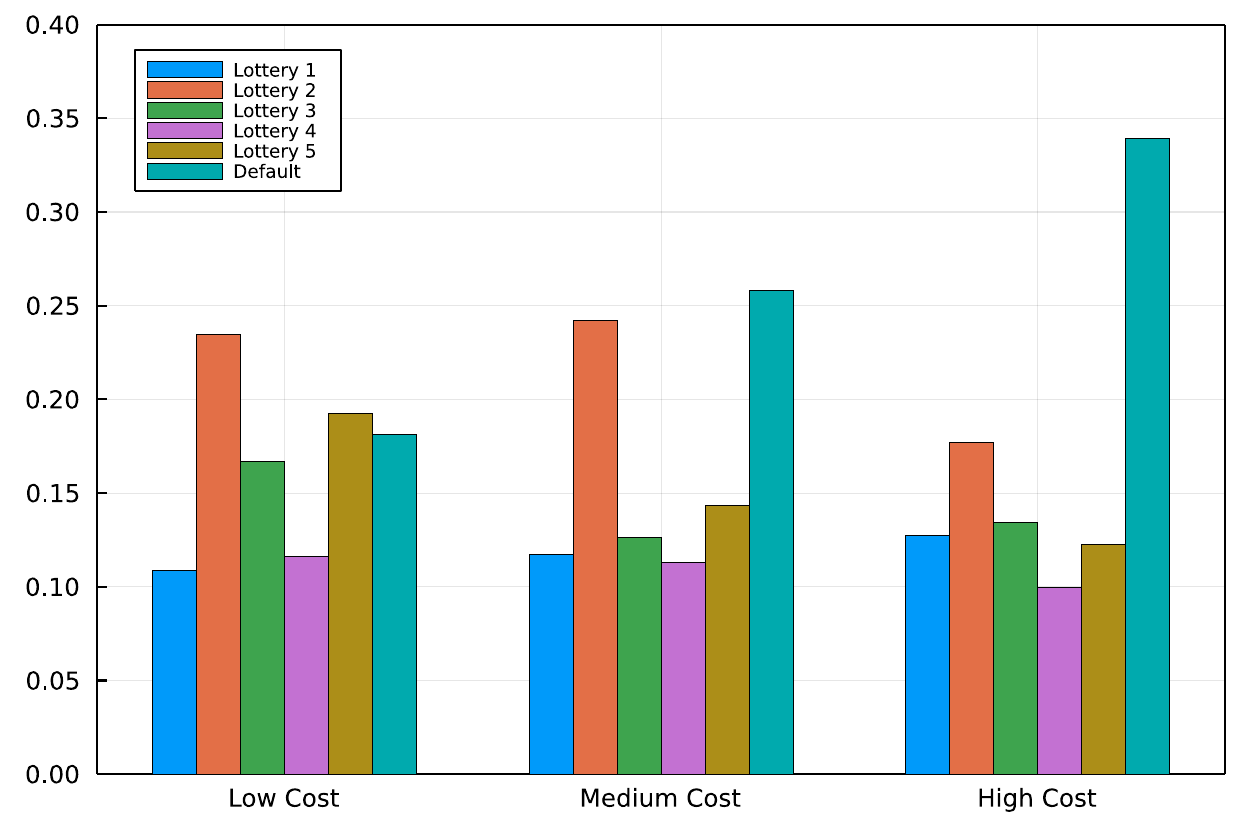}
	\caption{\textsc{Frequencies of Choices For Different Lotteries and Frames} } \label{fig:choices_cost}
\end{figure}
Figure~\ref{fig:choices_cost} also shows that the effect of the cost treatment is not homogeneous across alternatives. The choice frequency of lottery $1$ (the simplest to understand after the default) increases with the cost treatment; the cost does not have a significant impact on the probability of selecting lottery $4$; while it has a negative impact on the other lotteries.
\par
Overall, the probability of selecting any given lottery does not necessarily decrease with the size of the menu, suggesting that regularity is violated in our sample, as shown in Figure~\ref{choice_menu_size}. Indeed, we test $\rumo$ later and confirm that these violations are significant, since we reject $\rumo$ at the $5$ percent significance level.
\par 
Notice that, for a fixed frame/attention cost, some lotteries are harder to compute than others. For instance, lottery $2$ can pay $30$ or $10$ tokens, while lottery $3$ can pay $50$, 
$30$, $10$, or $0$. So for every frame, lottery $3$ is harder to compute than lottery $2$. Our attention-index framework allows for alternatives or lotteries with heterogeneous complexity (e.g., the simplest model $\mm$ has a lottery-specific attention parameter). For instance, the simplest lottery $1$ is picked more often as the cost rises, while the hardest to compute lottery $5$ displays the opposite pattern.  These observations confirm the presence of heterogeneous attention patterns.

\subsection{Evidence for Total Monotonic Attention}
Recall that, under the null hypothesis, the dataset is generated by an attention-index model with a link function that is totally monotone. Hence,  the frequency of choosing the outside option is equal to a monotone transformation of cumulative attention  associated with the attention-index  $\sum_{C\subseteq A} \eta(C)$:
\[
p(o,A)=m_{A}(\emptyset)=\varphi\left(\sum_{C\subseteq A} \eta(C),\eta(\emptyset)\right).
\]
 
In Figure~\ref{choice_menu_size}, we observe that the frequency of choice of the outside option \emph{is not} a decreasing function of the cardinality of the choice set. The latter is inconsistent with total monotonicity, and, thus, we may naively conclude that neither $\rumo$, $\la$, nor $\rcg$ can explain our dataset. However, as we confirm in the next section, this slight violation of monotonicity is an artifact of sample variability, and it is not statistically significant.  
\par 
Choice overload refers to the case when the propensity of not choosing any alternative (i.e., the probability of picking the default alternative) increases with the size of the choice sets \citep{iyengar2000choice}.  Our findings are informative on whether this effect, which may be present at the individual level, still matters at the population level. Note that neither $\rumo$ nor any model of limited consideration that we study (including \citealp{cattaneo2017random}) can rationalize choice overload.\footnote{\citet{chernev2015choice} provides a meta-data analysis of the determinants of choice overload.} 
Existing models of limited consideration are fundamentally at odds with choice overload, since one of the important reasons to form a consideration set is to simplify choice.\footnote{Other models of stochastic choice that are not models of limited consideration usually can accommodate choice overload. See, \citet{fudenberg2015stochastic, echenique2018perception, kovach2018imbalanced}, and \citet{Natenzon2018}.}
We find no statistical support for choice overload in our dataset. However, our choice set is of moderate size.

\subsection{Differentiating Between \texorpdfstring{$\la$}{LA} and \texorpdfstring{$\rcg$}{EBA}}
In Figure~\ref{choice_menu_size}, we also plot the frequencies of choice of the nondefault alternatives and find that there are violations of total monotonicity (e.g., lottery $3$).   This evidence suggests that $\rcg$ cannot describe this dataset because total monotonicity should hold for all alternatives \citep{block_random_1960,aguiar2017random}. However, this is not enough to conclude that $\rcg$ cannot describe the dataset because of sample variability. Nonetheless, we reject the null hypothesis that $\rcg$ can describe this dataset. In addition, we must  highlight that $\la$ is the only candidate that can accommodate the observed nonmonotonicity of the nondefault alternatives observed in this dataset. We confirm this insight in our formal testing section by showing that $\la$ does a good job at describing this dataset.
	
	\begin{figure}[ht]		
	\centering
		\includegraphics[width=1\textwidth]{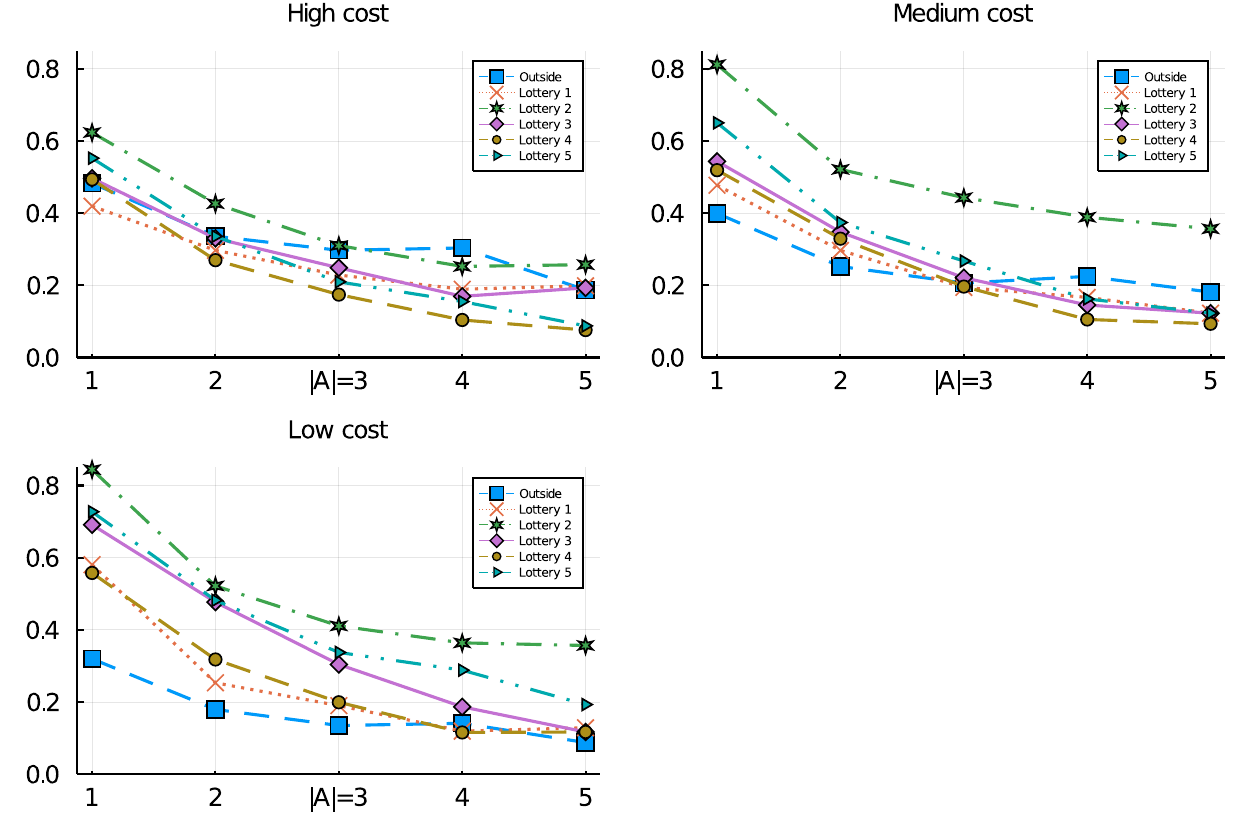}
		\caption{\textsc{Frequency of Choice for all frames as a function of the menu size}}\label{choice_menu_size}
	\end{figure}
The nonmonotonicity that we observe in Figure~\ref{choice_menu_size} is usually called \emph{attraction effect} \citep{huber1982adding}.  The attraction effect refers to a phenomenon when, as a new alternative is added to the choice set,  the probability of the existing items is boosted. Our findings support the presence of the attraction effect in our experimental sample.

\section{Testing Random Consideration Models}\label{section:testing}
In this section, we report the results of testing the ability of $\rumo$, $\la$, and $\rcg$ to describe our experimental data. We test these models without imposing any restrictions on preferences except stability over frames. 
Unless otherwise stated, the tested hypothesis is that, for a particular specification of our model (consideration set stochastic rule), there exists $(m,\pi)$ that is a $\rhrc$ representation for behavior.

\subsection{Econometric Testing}\label{subsec: econometric testing}
For every frame $f$, although $P_f$ is not observed, the realized choice frequencies $\hat{P}_f$ are. For every $A\in\M$ and $f\in F$, let $n_{f,A}$ denote the number of individuals in the sample that faced choice set $A$ and frame $f$, and let $\rand{a}_{i,f,A}$, $i=1,\dots,n_{A}$ be the observed choice of individual $i$ from choice set $A\cup\{o\}$ and frame $f$. We assume that the researcher observes a cross-section of observations (i.e., i.i.d. observations) for every menu and frame.\footnote{For a given frame, this is a standard stochastic choice dataset in the literature on limited consideration.} Then we define the estimated stochastic choice rule as 
\[
\hat P_f=(\hat p_f(a,\A))_{A\in\M,a\in A\cup\{o\}}.
\]
with $\hat p_f(a,\A)=n_{f,A}^{-1}\sum_{i=1}^{n_{f,A}}\Char{\rand{a}_{i,f,A}=a}$ for any $a\in A\cup\{o\}$. 
\par
Given the model of interest $\rc$ and the estimator of $P_f$, $\hat{P}_f$, we can compute the estimators of $m_f^{\rc}$ and $P_{f,\pi}^{\rc}$, $\hat m_f^{\rc}$ and $\hat P_{f,\pi}^{\rc}$, using Definitions~\ref{def:calibrated m} and~\ref{def:calibrated P}.\footnote{To compute $\hat P_{f,\pi}^{\rc}$, in our empirical application, we minimized the Euclidean distance between $\hat P_f$ and $\left(\sum_{D\subseteq A}\hat{m}_f^{\rc} p_{f,\pi}(a,D)\right)_{a\in A,A\in\mathcal{A}}$ subject to $p_{f,\pi}(a,A)\geq 0$, $\sum_{a\in D}p_{f,\pi}(a,D)=1$ for all $a$ and $A$, and $p_{f,\pi}(a,D)=0$ for all $D$ and $a\not\in D$.} 
Given the results of Corollary~\ref{cor:WRU characterization2}, a natural test statistic is
\begin{align*}
    \mathrm{T}_n=n\min_{v\in\Real^{d}_{+}}\norm{\hat g^{\rc}-Gv}^2,
\end{align*}
where $n=\min_f(\sum_{A}n_{f,A})$ is the smallest sample size across frames and $\hat g^{\rc}=\left(\left(\hat P_{f\pi}^{\rc}\right)_{f\in F}\tr,\left(\hat m_f^{\rc}\right)_{f\in F}\tr\right)\tr$.
\par
Let $\hat g^{\rc,*}_l$, $l=1,\dots,L$, be bootstrap replications of $\hat g^{\rc}$. Let  $\tau_n\geq 0$ be a tuning parameter and $\iota$ be a vector of ones of dimension $d$.\footnote{In our empirical application we conducted tests for different values of $\tau_n$ (e.g., $\tau_n=\sqrt{\dfrac{\log(\min_{f,A}n_{f,A})}{\min_{f,A}n_{f,A}}}$ following \citet{kitamura2018nonparametric}, and $\tau_n=0$). The results are qualitatively the same.} To compute the critical values of $\mathrm{T}_n$ we follow the bootstrap procedure proposed in \citet{kitamura2018nonparametric}:
\begin{enumerate}
    \item Compute $\hat\eta_{\tau_n}=Gv_{\tau_n}$, where $v_{\tau_n}$ solves
    \[
    n\min_{[v-\tau_n \iota/d]\in\Real^{d}_{+}}\norm{\hat g^{\rc}-Gv}^2;
    \]
\item Compute the bootstrap test statistic
\[
\mathrm{T}^{*}_{n,l}=n\min_{[v-\tau_n \iota/d]\in\Real^{d}_{+}}\norm{\hat g^{\rc,*}_l-\hat g^{\rc}+\hat\eta_{\tau_n}-Gv}^2,\quad l=1,\dots,L;
\]
\item Use the empirical distribution of the bootstrap statistic to compute critical values of $\mathrm{T}_n$.
\end{enumerate}
For a given significance level $\alpha\in (0,1/2)$, the decision rule for the test is ``reject the null hypothesis if $\mathrm{T}_n>\hat c_{1-\alpha}$'', where $\hat c_{1-\alpha}$ is an $(1-\alpha)$-quantile of the empirical distribution of the bootstrap statistic. 
\par
We would like to conclude this section by observing that we can test the model conditional on additional observables (e.g., gender, income brackets, and education level). For discrete (or discretized) covariates one just needs to perform the test for a subgroup of the population.

\subsection{Survival Race the \texorpdfstring{$\rhrc$}{Link-Behavioral}-rule v.s. \texorpdfstring{$\rumo$}{Random Utility Model}: Stability of Preferences}\label{subsection:stability_pi}
Without frame variation, many models of consideration are empirically indistinguishable from $\rumo$. For instance, if a dataset is consistent with $\rcg$ or $\mm$, for a fixed frame, then it is also consistent with $\rumo$. However, these models and $\rumo$ will typically recover a distinct distribution of preferences. Varying frames, we can test whether the distribution of preferences remains the same across frames. 
\par
$\rhrc$-model assumes that the distribution of preferences in the population is independent of the consideration rule.  In our experiment, the choice sets faced by any subject are exactly the same for the three consideration cost treatments. Given our pay-at-random incentives scheme, choices from each choice set can be considered as i.i.d. draws from the underlying random utility distribution under the null hypothesis of stochastic rationality. Therefore, the independence assumption together with our experimental design imply that if one of the $\rhrc$ theories describes the behavior of the high-cost treatment, it must also describe the behavior of the low-cost treatment. That is, if the independence assumption holds, then the distribution of preferences, $\pi$, should be invariant to changes in consideration costs for theories that we cannot reject. We check the validity of the different models of interest under this preference stability restriction. 
\par
We apply the procedure described in Section~\ref{subsec: econometric testing} to test whether the $\rumo$, $\la$, and $\rcg$ models can explain the data with the restriction that the distribution of preferences remains frame-independent (i.e., consequentialism).
The results of testing are presented in Table~\ref{results_joit_pref}. In this table, we report the values of the test statistic and the corresponding p-values coming from the bootstrap distribution ($1000$ bootstrap replications for every test statistic were conducted) for different models.\footnote{The p-value is interpreted as the probability of observing a realization of the test statistic that is above the one that is actually observed due to sample variability, if the null hypothesis is indeed correct. Then, the smaller the p-value is, the more evidence the researcher has to reject the hypothesis of the validity of a given model.}
First, we strongly reject $\rumo$ at any reasonable significance level. In other words, for $\rumo$ we reject the hypothesis that the same distribution of preferences can rationalize behavior across consideration cost treatments. In contrast, we cannot reject the $\la$ model at any standard significance level. In addition, we reject the hypothesis that $\rcg$ explains the population behavior under preference stability. Note that we can discriminate between $\rumo$ and $\rcg$ because of variation in frames--$\rcg$ is more general than $\rumo$ because it allows for flexible attention per frame. So the rejection of $\rcg$ with stable preferences does not follow from the rejection of $\rumo$. Taken together, our results show that our experimental subjects behave as if they are maximizing their preferences given a consideration set that follows the $\la$ restriction. 

\begin{table}
\begin{center} 
\caption[Consistency preferences]{\textsc{Testing Results under Preference Stability}  } \label{results_joit_pref} 
\centering
\begin{tabular}{@{\extracolsep{2pt}} |l cc  |  @{}}
\cline{1-3}
\multicolumn{3}{|c|}{}\\
\textsc{Model}&\textsc{$\mathrm{T}_n$}&\textsc{p-value}\\ [2ex]
\cline{1-3} 
\multicolumn{3}{|c|}{}\\
$\rumo$ &	3231.59 & <0.001\\[2ex]
$\la$ &24959.06 & 0.524\\[2ex]
$\rcg$  &24840.23& 0.001\\
\multicolumn{3}{|c|}{}\\ \cline{1-3}
\multicolumn{3}{l}{\footnotesize{Notes: Number of bootstrap replications=1000.}}
\end{tabular}
\end{center} 
\end{table}

\subsection{Discussion}\label{subsection:discussion}
Our findings strongly support the hypothesis that the population behaves as if it is consistent with the $\la$ model of limited consideration and has a stable distribution of preferences across frames.\footnote{We also rejected the hypothesis of whether the $\la$ model with homogeneous preferences can explain the data. See Appendix~\ref{app: homogeneous LA}.} All frame effects observed in our descriptive analysis are fully captured by the variation in the random consideration rule that changes conditional on the frame. In contrast, the traditional $\rumo$ fails to describe the population behavior. To confirm that our testing procedure has power against $\la$, in Appendix~\ref{app: MC simulations}, we access the performance of our procedure using Monte Carlos simulation. In particular, we show that our test can reject the false null hypothesis of data being consistent with the $\la$ model with high frequency in finite samples that are comparable to our experiment.
\par
We highlight that our analysis cannot exclude the possibility that other models of behavior could also explain the population behavior in our experiment. We have only established that the population behavior is \emph{as if} it is consistent with a $\rhrc$-rule. 
\par
Although we do not impose any restrictions on preferences, e.g expected utility, our results relate to the work of \citet{freeman2018eliciting}. They provide an alternative mechanism for the selection of a riskless lottery (default) over dominant risky choices from pairwise comparisons, when binary choice sets are presented as lists. They propose a theoretical explanation of the choice of the riskless choice with a model of reference dependence. The class of reference dependence models used by these authors is a special case of utility maximization. Recall that we find evidence against $\rumo$ in our experiment, thus, ruling out \citet{freeman2018eliciting} mechanism for our environment with costly consideration. In addition, in our experimental design, subjects are not required to choose from lists nor are restricted to pairwise comparisons. For an extended discussion of the role of misperception see Appendix~\ref{subsectionapp:imperfect_perception}.
\par 
We have maintained the assumption that the default alternative is also the worst alternative for both the $\la$ and $\rcg$ models. However, $\rumo$ allows the default to be ranked arbitrarily. Hence, the main findings that $\la$ explains the dataset and $\rumo$ fails to do the same are robust to this assumption. We leave it as an open question whether $\rcg$ can explain this dataset if this assumption is relaxed. Nevertheless, we believe that this assumption is reasonable in our experimental setup.
\par
We have done our empirical analysis without conditioning on observable heterogeneity (e.g., age or gender). Attention and preferences may differ across different demographic groups. Methodologically, our tools can be applied after conditioning on observable heterogeneity, as explained in \citet{kitamura2018nonparametric}. The study of consideration set rules and their relation to demographics is beyond the scope of this paper.
\par 
It is noteworthy that if individuals have convex risk preferences to mix between lotteries, repetition of discretized choice tasks may make it hard to estimate risk preferences.  \citet{feldman2020revealing} show that subjects who mix between lotteries in convex budgets sets are more likely to randomize choices in a repeated discretized task. Our experiment that uses different cost-treatments and randomizes the order of lotteries in presentation may alleviate such concerns compared to previous experiments in which the same set of lotteries in a fixed order was repeatedly presented to subjects.\footnote{We thank an anonymous referee for pointing out this issue.} 
\par 
We finish this section by discussing our model and our findings in relation to Rational Inattention (RI) models. \citet{caplin2016rationalconsideration} shows that rational inattentive DMs form (deterministic) consideration sets. Generally, RI primitives cannot be point-identified with standard stochastic choice datasets. Nonetheless, RI models may still have testable implications in standard stochastic choice datasets. In Appendix~\ref{subsectionapp:RI}, we show that a representative RI DM is compatible with deterministic consideration sets (i.e, the presence of zero probability of choice), which is not supported in our data. The case of a population of heterogeneous rational inattentive DMs and the aggregation of such behavior in the population is left for future research. 

\section{Conclusion} \label{section:conclusion}
We have designed and implemented a novel experiment with a large sample that allowed us to statistically discern among competing models of population behavior. By exogenously varying choice sets and the frames induced by the cost of considering alternatives, we can disentangle two sources of stochastic behavior: limited consideration and preference heterogeneity. We use this novel dataset to test $\rumo$ and two models of limited consideration, $\la$ and $\rcg$. 
\par
These models provide testable implications on choices that uniquely identify the stochastic consideration set rule from data. By calibrating consideration given the theory, we show that testing the $\rhrc$-model can be cast into \citet{kitamura2018nonparametric} framework for testing $\rumo$. That is, we show that there exists a stochastic rule (computed from data) that is $\rumo$ if and only if observed choices are generated by a population of individuals consistent with the $\rhrc$-model.
\par
We provide evidence against classical $\rumo$, since consideration costs are binding for some individuals in the population. In contrast, we find support for the $\la$ model with heterogeneous preferences. Crucially, we cannot reject that the distribution of preferences implied by $\la$ is the same across all attention frames. This means that once we disentangled attention and preferences under $\la$, the recovered distribution of preferences does not change with the frame. 

\appendix
\section{Proofs}\label{sectionapp:proofs}
\subsection{Proof of Lemma~\ref{lemma:empiricalcontent}}
We define $m_A(\{a\})=p(a,A)$, and $m_A(D)=0$ for all $D\subseteq A$, $D\neq\{a\}$. Let $\tilde{\pi}\in\Delta(R(X))$ be the uniform distribution. The pair $((m_A)_{A\in\mathcal{M}},\tilde{\pi})$ is a $\hrc$-rule. We now prove that it generates any data $P$. By definition if $P$ can be generated by a $\hrc$-rule, then we have that 
\[ 
p(a,A)=\sum_{D\subseteq A}m_A(D)\sum_{\rr\in R(X)}\tilde{\pi}(\rr)\Char{a\rr b,\forall b\in D}
\]
for all $A$ and $a\in A$.
Rearranging and replacing the choice of $\tilde{\pi}$ in the above equation we get that
\[
\sum_{D\subseteq A}m_A(D)\sum_{\rr\in R}\tilde{\pi}(\rr)\Char{a\rr b,\forall b\in D}=\frac{1}{\abs{R(X)}}\sum_{\rr\in R(X)}[\sum_{D\subseteq A}m_A(D)\Char{a\rr b,\forall b\in D}]\,.
\]

For given $\rr$ and $m_A(\{a\})=p(a,A)$, we have 
\[
\sum_{D\subseteq A}m_A(D)\Char{a\rr b,\forall b\in D}=p(a,A)\Char{a\rr a}=p(a,A)
\]
because $\rr$ includes the diagonal $a\rr a$ for all $a\in X$. 

This implies that 
\[
\frac{1}{\abs{R(X)}}\sum_{\rr\in R(X)}[\sum_{D\subseteq A}m_A(D)\Char{a\rr b,\forall b\in D}]=p(a,A)
\]
given that
\[
\frac{1}{\abs{R(X)}}\sum_{\rr\in R(X)}[\sum_{D\subseteq A}m_A(D)\Char{a\rr b,\forall b\in D}]=\frac{1}{\abs{R(X)}}\sum_{\rr\in R(X)}p(a,A)=p(a,A)\,. 
\]

\subsection{Proof of Theorem~\ref{thm:calibration of m}}
(i) implies (ii).
A complete stochastic choice rule $P$ is a $\hrc$-rule if there exists a pair $(m,\pi)$ such that 
\[
p(a,\A)=\sum_{\D\subseteq \A}m_\A(\D)\sum_{\rr\in R(X)}\pi(\rr)\Char{a \rr b,\:\forall\: b\in \D}\,,
\]
for all $a\in X$ and $\A\in\M$, where we exchanged the summation operator with respect to the consideration sets and the linear orders exploiting independence. 

Note that we can write the probability of the default alternative as $p(o,\A)=1-\sum_{a\in \A}p(a,\A)$. This implies that
\[
1-p(o,\A)=\sum_{\D\subseteq \A}m_\A(\D)[\sum_{a\in \A}\sum_{\rr\in R(X)}\pi(\rr)\Char{a \rr b,\:\forall\: b\in \D}]\,,
\]
where the summation operator with respect to the items $a\in \A$ can be exchanged with the summation over consideration sets. This is possible because the latter summation does not depend on the items $a\in \A$. 

Now, we notice that $\sum_{a\in \A}\sum_{\rr\in R(X)}\pi(\rr)\Char{a \rr b,\:\forall\: b\in \D}=1$ for all $D\subseteq A$.
This implies that the default probability does not depend on the distribution of preferences and can be written in terms of the cumulative distribution of the consideration set distribution:

\[
p(o,\A)=1-\sum_{\D\subseteq \A,\D\neq \emptyset}m_\A(\D).
\]

We let the capacity $\varphi^*:2^X\to [0,1]$ be defined by $\varphi^*(A)=p(o,\A)$.
\par
The fact that $\eta=\eta^\rc$ under the correct specification of the link function follows from our monotonicity assumptions and the existence of a unique Mobius inverse of the mapping $v(\cdot)=\sum_{C\subseteq \cdot}\eta(C)$ \citep{shafer1976mathematical,chateauneuf1989some}. We provide specific derivations for each of the models of interest in this paper, to connect them with the existing literature, but they follow directly from the general $\eta^\rc$ formula. 
\par
For given $\rc\in\{\la,\mm,\rcg,\rum\}$ and $P$:

\begin{itemize}

\item If $m\in \mathcal{M^{\la}}$, then $m_A(D)=\frac{\eta(D)}{\sum_{C\subseteq A}\eta(C)}$ for some $\eta\in \Delta(2^X)\cap \Real_{++}$. 

This means that $\frac{\varphi^*(X)}{\varphi^*(A)}=\sum_{D\subseteq A}\eta(D)$. Then by \citet{shafer1976mathematical} it must be that 
\[
\eta(D)=\sum_{B\subseteq D}(-1)^{\abs{D\setminus B}}\frac{\varphi^*(X)}{\varphi^*(B)}=\sum_{B\subseteq D}(-1)^{\abs{D\setminus B}}\frac{p(o,X)}{p(o,B)}\,;
\]

\item If $m\in \mathcal{M}^{\mm}$, then $m_A(D)=\frac{\eta(D)}{\sum_{C\subseteq A}\eta(C)}$ for some $\eta\in \Delta(2^X)\cap \Real_{++}$ with,

\[
\eta(D)=\prod_{a\in X\setminus D}\left(1-\gamma(a)\right)\prod_{b\in D}\gamma(b),
\]
and $\gamma:X\rightarrow (0,1)$. 
This implies by simple computation that 
\[
\gamma(a)=1-\frac{\varphi^*(A)}{\varphi^*(A\setminus{\{a\}})}=1-\frac{p(o,A)}{p(o,A\setminus{\{a\}})} 
\]
for some $A\in \M$ that contains $a$;

\item If $m\in \mathcal{M}^{\rcg}$, then $m_A(D)=\sum_{C:C\cap A=D}\eta(C)$ for some $\eta\in \Delta(2^X)$. Then
\[
\varphi^*(A)=\sum_{D\cap A\neq \emptyset}\eta(D).
\]
Using \citet{shafer1976mathematical} and \citet{chateauneuf1989some} we conclude that
\[
\eta(D)=\sum_{A\subseteq D:D\in \M}(-1)^{\abs{D\setminus{A}}}(1-\varphi^*(X\setminus{A}))=\sum_{A\subseteq D:D\in \M}(-1)^{\abs{D\setminus{A}}}(p(o,X\setminus{A}));
\]

\item If $m$ is $\rum$, then obviously $m_{A}(D)=\Char{A=D}$.
\end{itemize}

To establish that $m=m^{\rc}$ for given $\rc\in\{\la,\mm,\rcg,\rum\}$ and $P$, we exploit the uniqueness of $m$, which is a consequence of the invertibility of the Mobius transform and the completeness of $P$. In particular, if $(m,\pi)$ and $(m',\pi)$ represent the same $P$ then it must be that $m'=m$ for the cases of $\rc\in\{\la,\mm,\rcg,\rum\}$. 
To see that this is true recall that if $P$ is a $\rhrc$-rule with $(m,\pi)$, then 
$1-\sum_{\D\subseteq \A,\D\neq \emptyset}m_\A(\D)=\varphi^*(A)$.
This is exactly the same for the case where there is homogeneity in the preferences such that there is a linear order $\succ\in R(X)$ such that $\pi(\succ)=1$. Since this equivalence does not depend on the distribution of preferences and due to the completeness of the dataset, we can use this fact to apply known results from the consideration set literature regarding the uniqueness of $m$.

Now, by the Mobius inverse, it follows that 
\[
\eta^{\la}(D)=\sum_{B\subseteq D}(-1)^{\abs{D\setminus B}}\frac{p(o,X)}{p(o,B)}\,,
\] 
for all $D\in 2^X$. 
In particular,
\begin{itemize}
\item By Theorem~$3.1$ in \citet{brady2016menu}, it must be that $m$ is uniquely identified by 
\[
m_A(D)=\frac{\eta(D)}{\sum_{C\subseteq D}\eta(D)}\,,
\]
where $\eta\in \Delta(2^X)\cap \Real_{++}$ follows from the requirement that $\sum_{B\subseteq D}(-1)^{\abs{D\setminus B}}\frac{p(o,X)}{p(o,B)}>0$ for all $D\in 2^X$.

\item Given $\gamma^{\mm}(a)=1-p(o,{a})\in (0,1)$ for all $a\in X$ (which is well-defined by the completeness of $P$) and $\eta^{\mm}(D)=\prod_{a\in X\setminus D}\left(1-\gamma^{\mm}(a)\right)\prod_{b\in D}\gamma^{\mm}(b)$, it follows that $m$ is uniquely identified by
\[
m_A(D)=\prod_{a\in D}\gamma^{\mm}(a)\prod_{b\in A\setminus D}(1-\gamma^{\mm}(b))\,,
\]
for all $A\subseteq D$. Note that $\prod_{b\in \emptyset}\gamma^{mm}(b)=1$ by convention. Also observe that uniqueness follows from Theorem~$3.3$ in \citet{brady2016menu} since the $\mm$ restriction is a special case of the $\la$ restriction.

\item Given $\eta^{\rcg}(D)=\sum_{A\subseteq D:D\in \M}(-1)^{\abs{D\setminus{A}}}(1-p(o,X\setminus{A}))\geq 0 $ it follows by Theorem~$1$ in \citet{aguiar2017random} that $m$ is uniquely identified by
\[
m^{\rcg}_A(D)=\sum_{C:C\cap A=D}\eta^{\rcg}(C)\,,
\]
for all $D\subseteq A$, where $D\neq \emptyset$ and $m_A(\emptyset)=1-\sum_{D\subseteq A,D\neq \emptyset}m_A(D)$. 
\item The case of $\rum$ is trivial. 
\end{itemize}

\subsection{Proof of Theorem~\ref{thm:WRU characterization}}
(i) implies (ii).
If $P$ is a $\rhrc$-rule then by Theorem~\ref{thm:calibration of m}, under conditions (i) and (ii), it must be that 
\[
p^{\rc}_{\pi}(a,\A)=\dfrac{p_{m,\pi}(a,\A)-\sum_{C\subset\A}m^{\rc}_\A(C)p^{\rc}_{\pi}(a,C)}{m^{\rc}_{\A}(\A)}\,,
\]
where $p_{m,\pi}(a,\A)=\sum_{D\subseteq A}m_A^{\rc}(D)[\sum_{\succ\in R(X)}\pi(\succ)\Char{a\succ b\forall b\in D}]$. Following the recursive formula, we can show that
\[
p^{\rc}_{\pi}(a,\A)=\dfrac{m^{\rc}_{\A}(\A)[\sum_{\succ\in R(X)}\pi(\succ)\Char{a\succ b\forall b\in A}]}{m^{\rc}_{\A}(\A)}=\sum_{\succ\in R(X)}\pi(\succ)\Char{a\succ b\forall b\in A}\,.
\]
This implies that $P^{\rc}$ is a $\rum\text{-}\hrc$-rule.

(ii) implies (i).
Under conditions (i) and (ii), the fact that
\[
p^{\rc}_{\pi}(a,\A)=\dfrac{p_{m,\pi}(a,\A)-\sum_{C\subset\A}m^{\rc}_\A(C)p^{\rc}_{\pi}(a,C)}{m^{\rc}_{\A}(\A)}
\]
implies that for all $A\in \M$ and all $a\in A$,
\[
p(a,A)=\sum_{D\subseteq A}m^{\rc}_A(D)p^{\rc}_{\pi}(a,D)\,. 
\]
If $P^{\rc}$ is a $\rum\text{-}\hrc$-rule, then it implies that there exists $\pi \in \Delta(R(X))$ such that
\[
p^{\rc}_{\pi}(a,\A)=\sum_{\succ\in R(X)}\pi(\succ)\Char{a\succ b\forall b\in A}\,.
\]
Hence, $P$ is a $\rhrc$-rule. In fact, for all $A\in \M$ and all $a\in A$, it must be that the pair $(m^{\rc},\pi)$ generates the dataset $P$:
\[
p(a,A)=\sum_{D\subseteq A}m^{\rc}_A(D)\sum_{\succ\in R(X)}\pi(\succ)\Char{a\succ b\forall b\in A},.
\]

\subsection{Proof of Theorem~\ref{thm:identification}}
We first prove that if $P$ is described by $(m,\pi)$ and $(m',\pi')$, then it must be that $m=m'$. This follows from \citet{chateauneuf1989some}. In particular, \citet{brady2016menu} shows the identification results for $\rc=\la$, while \citet{aguiar2017random} provides identification results for $\rc=\rcg$. For $\rc=\mm$ the result holds trivially.

Fixing $m$, if $P$ is described by both $(m,\pi)$ and $(m,\pi')$, then 
\[
p_{\pi}^{\rc}(a,A)=\dfrac{p(a,\A)-\sum_{C\subset\A}m^{\rc}_\A(C)p^{\rc}_{\pi}(a,C)}{m^{\rc}_{\A}(\A)}\,,
\]
and 
\[
p_{\pi'}^{\rc}(a,A)=\dfrac{p(a,\A)-\sum_{C\subset\A}m^{\rc}_\A(C)p^{\rc}_{\pi'}(a,C)}{m^{\rc}_{\A}(\A)}\,,
\]
for all $a\in A$ and nonempty $A\subseteq X$, which follows from Definition~\ref{def:calibrated P}. By condition (ii), $m_{A}^{\rc}(A)>0$ and using the recursive definitions above for binary sets, we can see that $p_{\pi}^{\rc}(a,\{a,b\})=p_{\pi'}^{\rc}(a,\{a,b\})$ for any $a,b\in X$. For a fixed $m$ the recursive formula leads to the equivalence $p^{\rc}_{\pi'}=p^{\rc}_{\pi}$.

\section{Sensitivity Analysis for the Default\label{appendix:sensitivitydefault}}
One key assumption in our setup is that the outside option is only picked under full consideration if the choice set only contains the outside option. In our notation, we write this as $p_{\pi}(o,\emptyset)=1$ and $p_{\pi}(o,A)=0$ for all $A\neq \emptyset$.  
In this section, we relax this assumption to allow the probability of choosing the default to satisfy $p_{\pi}(o,\emptyset)=1$ and $p_{\pi}(o,A)=\frac{e}{\abs{A}}$ with $e\in [0,1)$ for all $A\neq \emptyset$. The sensitivity parameter $e$ is interpreted as the fraction of DMs that choose the dominated default even when there are other available alternatives. This is a violation of the assumption that the default is the worst alternative. Then this implies that under the null of consistency with the $\hrc$-rule the probability of choosing the default is:
\[
p(o,A)=m_A(\emptyset)+\frac{e}{\abs{A}}(1-m_A(\emptyset)).
\]
This assumption is compatible with $\rumo$ and implies that the probability of $o$ being the best alternative in all linear orders is constant across menus under full consideration.\footnote{Indeed, this assumption allows the default to be chosen when compared to other alternatives. This is achieved by putting a mass (equal to $e/\abs{A}$) on the event that $o$ is the first among all alternatives in a given menu. This is a restriction on the distribution of preferences.} Notice, that under this assumption we can calibrate $m_A(\emptyset)$ for all $A\in \M$:
\[
m_A(\emptyset)=\frac{p(o,A)-\frac{e}{\abs{A}}}{\left(1-\frac{e}{\abs{A}}\right)}.
\]
Then for a given sensitivity parameter $e\in[0,1)$, we can calibrate the empirical attention-index $\eta^{\rc}$ without any changes.
Next, we can compute $p_{\pi}$ given the calibrated $m^{\rc},$ for $A\cup \{o\}$, and we can test $\rumo$ here using the tools we have developed. In our current empirical results, since the $\la$ model passes for $e=0$, it is unnecessary to do this sensitivity analysis.    

\section{Comparison with Models of Stochastic Choice}\label{app:section_models}
In this appendix we analyze the connection between the three consideration-mediated choice theories discussed in this paper and models that allow for stochastic behavior exclusively in preferences or in consideration. 

\subsection{Comparison with RUM}\label{subsectionapp:comparison_RUM}

As explained previously, randomness arising from limited consideration as in $\rcg$ and $\mm$ can be rationalized under the umbrella of random utility for a fixed frame. However, $\la$ allows for behavior that is inconsistent with regularity. Therefore $\la$ is not nested in $\rumo$. By construction our model $\rhrc$ generalizes $\rum$ by allowing for independent variation in choices due to limited consideration. In particular, $\rhrc$ is $\rumo$ defined over $X$ (what we call, equivalently, $\rum$) when the stochastic choice rule is such that $m_A(D)=\Char{D=A}$. We call this model $\rum$. 

Moreover, the $\rhrc$ model is more general than $\rumo$. This follows from the analysis in previous sections. In particular, fixing preferences, $\pi(\succ_i)=\Char{\succ_i=\succ}$ for $\succ_i \in R(X)$, $\rhrc$ with $\rc=\la$ reduces to original $\la$ model by \citet{brady2016menu}, and therefore potentially inconsistent with $\rumo$. Of course, when we add frame variation we can distinguish between $\rcg$ and $\rumo$.

\subsection{Comparison to RAM } \label{subsectionapp:comparison_RAM}

\citet{cattaneo2017random} extends many theories of consideration by proposing the Random Attention Model (RAM). The authors allow for random consideration maps in the context of limited attention models. RAM abstracts away from the particular consideration-set-formation rule by considering a class of nonparametric random attention rules. The authors acknowledge that \emph{RAM is best suited for eliciting information about preference ordering of a single decision-making unit when her choices are observed repeatedly}, which justifies the preference homogeneity assumption in their setting. 

Many of the canonical models of limited attention proposed in the literature satisfy the Monotonic Attention property of \citet{cattaneo2017random}. For instance, RAM nests $\la$, $\mm$ and $\rcg$ without preference heterogeneity among other salient models of consideration sets. Additionally, RAM is a strict generalization of $\rumo$. However, our $\rhrc$ is not nested in RAM, see \citet{cattaneo2017random} for a complete description of its relationship to the literature. 

Here we show that, in the presence of preference heterogeneity RAM may fail to rationalize behavior that can be explained by $\rhrc$. First, we formally define the restrictions imposed by RAM.

RAM imposes a \textit{ monotonic attention } restriction on consideration rules: the probability of paying attention to a particular subset does not decrease when the total number of possible consideration sets decreases. Formally,

\begin{definition}[Monotonic Attention] For any $a\in A\setminus D$, $m_A(D)\leq m_{A\setminus \{a\}}(D)$ \label{defn:monotonic_att}
\end{definition}

Moreover, \citet{cattaneo2017random} provides a characterization of the model in terms of the revealed preference information inferred from data. Formally,

\begin{definition} [Revealed Preference, RAM]
 Let $p$ be a RAM. Define $P_R$ as the transitive closure of $P$ defined as 
\[
aPb \text{ if there exists } A\in \M \text{ with }a,b\in A \text{ such that }p(a,A)>p(a,A\setminus \{b\})\,.
\] 
Then $a$ is revealed preferred to $b$ of and only if $a~P_R~ b$. 
\end{definition}

Then, a choice rule has a RAM representation if and only if $P_R$ has no cycles. The following example of a $\rhrc$-rule, which results from a $m\in\mathcal{M}^{\la}$ for two linear orders $\succ_1$ and $\succ_2$ with $\pi(\succ_i)=0.5$ with $i=1,2$, cannot be generated by RAM.

\begin{example}[RAM violation]\label{example_RAM}
Let $X=\{a,b,c\}$ and consider a $\la$ model for the random consideration set probability measure with $\eta(D)$ given as in Table~\ref{table_ex_RAM}. Moreover, consider two preference relations $\succ_1$ such that $a\succ_1 b\succ_1 c$, and $\succ_2$ such that $c\succ_2 b\succ_2 a$.\footnote{In our experiment, this preference heterogeneity can be explained by heterogeneity in risk aversion. For example, let $a\equiv l_4$, $b\equiv l_3$, $c\equiv l_2$, and assume that DMs are expected-utility-maximizers with CRRA Bernoulli utilities. Then $a\succ b\succ c$ for individuals that are risk-neutrals ($\sigma=0$), while $c\succ b\succ a$ for risk averse individuals ($\sigma>0.5$). \citet{holt2002risk} finds that these types are common in their experiment across payment schemes. } The resulting probabilistic choice rule is generated by a $\la\text{-}\hrc$ by construction. However, it cannot be rationalized by RAM since both $aPb$ and $bPa$ (i.e., $p(a,\{a,b,c\})>p(a,\{a,c\})$ and $p(b,\{a,b,c\})>p(b,\{b,c\})$). This means that a $\la\text{-}\hrc$-rule allows cycles of the revealed preference relation $P$, which is ruled out by RAM. 

\begin{table}
	\begin{center}
\caption[Example~\ref{example_RAM}.]{\textsc{Example~\ref{example_RAM}} Stochastic choice rule and random consideration set probability. $p$ is consistent with $\la\text{-}\hrc$ but cannot be generated by RAM.} \label{table_ex_RAM}
		\begin{tabular}{@{\extracolsep{3pt}}|c| c c c c c c c c| } 		\cline{1-9} 	
		&&&&&&&&\\
 & $\{a,b,c\}$& $\{a,b\}$& $\{a,c\}$ & $\{b,c\}$& $\{a\}$& $\{b\}$& $\{c\}$& $\emptyset$\\ \cline{1-9}
  &&&&&&&&\\
$a$&0.305&0.339&0.157&&0.208&&&\\
$b$&0.250&0.339&&0.227&&0.208&&\\
$c$&0.255&&0.300&0.341&&&0.345&\\
$o$&0.190&0.322&0.543&0.432&0.792&0.792&0.655&1\\ \cline{1-9} 
&&&&&&&&\\
$\eta(D)$&0.20	&0.30	&0.01&	0.10&	0.05	&0.05&	0.10&	0.19\\ [1ex]\cline{1-9}
			\end{tabular} 
		\end{center} 
	\end{table}
\end{example}

\subsection{Consideration Cost and Imperfect Perception}\label{subsectionapp:imperfect_perception}

One possible concern with our design is that DMs consider an alternative but misperceives the attributes (i.e., computes the wrong utility). We must point out that this concern applies broadly to any experimental design in which subjects have a nontrivial cognitive task. The following analysis assumes that the consideration cost is fixed. 
First, we need some preliminaries. For a given distribution of preferences $\pi \in \Delta(R(X))$, with perfect perception, there exists a random utility array $\rand{u}=(\rand{u}_a)_{a\in A}$ supported on $\Real^{\abs{A}}$ such that for a given menu of alternatives $A\in\M$: 
\[
\Prob{\rand{u}_a>\rand{u}_b,\:\forall b \in A\setminus{\{a\}}}=\pi(\succ\::\:a\succ b,\: \forall b\in A\setminus{\{a\}}).
\]
Now, miss-perception of any alternative $a\in A$ can be represented by another (possibly wrong) random utility variable $\rand{w}_{a}$ supported on the reals. We let $\rand{w}=(\rand{w}_{a})_{a\in A}$ be the array of such random variables. This random variable represents the subjective value that the DMs assigns to alternative $a$ given her own perception of the item. Hence, $\rand{w}$ and $\rand{u}$ may be different (even when they may be correlated). Without loss of generality, we can define miss-perception as:
\[
\rand{e}_{a}=\rand{w}_{a}-\rand{u}_a,
\]
for all $a \in A$ (and $\rand{e}=(\rand{e}_a)_{a\in A}$).\footnote{Note that in our experimental design, menus are randomly assigned to a DM. In addition, the presentation of each alternative remains the same across menu, conditional on the cost of consideration. Then, it must be that the distribution of miss-perception (by-design) is the same across menus.}  
\par

Under the assumption of independence of preference and attention. The only implication of miss-perception is that subjects' behavior will be governed not by $\pi$ but rather by a different distribution of preferences $\pi_e$ such that:
\[
    \pi_e(\succ:a\succ b ~\forall b \in A\setminus{\{a\}})=\Prob{\rand{u}_a+\rand{e}_a> \rand{u}_b+\rand{e}_b~\forall b \in A\setminus{\{a\}}}.\footnote{Where $\pi_e(\succ:\text{property})$ denotes the cumulative probability of preferences that have a certain property.}
\]
In other words, the population of DMs behavior captured by $P$ will still be represented by a $\rhrc$ model with $(\pi_e,m)$ instead of the true $(\pi,m)$. This means that our design is robust to any arbitrary miss-perception error, in terms of the validity of our conclusions about how good are the different models to describe the population. 
\par 
The only possible problem induced by miss-perception of alternatives is that we may lose the capacity to identify the true distribution of preferences. This possibility again is unavoidable in any experimental design that has a cognitive task. Nonetheless, this possibility is testable in our framework. In particular, if miss-perception exists, it must depend on the cognitive cost. Hence, we have the triple $(\rand{e}_H,\rand{e}_M,\rand{e}_L)$ that represents the miss-perception random array for the high, medium and low cost, respectively. 
Then, the distribution of preferences for any $\rhrc$ model must not be stable across attention treatments with corresponding $(\pi_{e_H},\pi_{e_M},\pi_{e_L})$ distribution of preferences (that differ among costs). 
\par
However, we cannot reject the null hypothesis that $\la$ has a stable distribution of preferences among the different cost distributions (i.e., $\pi_{e_H}=\pi_{e_M}=\pi_{e_L}=\pi$). In that sense, there is no evidence that miss-perception is important in our design.  

\subsection{Relation with Rational Inattention Models}\label{subsectionapp:RI}
Rational inattention (RI) models have recently gained a lot of interest to model situations when choice is hard. However, RI models usually need very rich datasets to be indentified/tested. That is, generally they cannot be identified with standard stochastic choice datasets. In that sense, we cannot do a full comparison between RI models and our approach, since the dataset requirements are different. However, we can derive some implications of RI behavior for our dataset.
\par
RI is a model for individual behavior. To the best of our knowledge the aggregate implications of this model are unknown. Hence we will focus on comparing our approach to a case of a representative RI behavior. The problem of the representative RI DM is to choose the best possible alternative from a choice set. She has a prior $\mu$ over the true value of alternatives, $V=(v_k)_{k\in X\cup \{o\}}$ , with $\mu \in \Delta(V)$. In response to the information structure, the RI DM chooses her optimal information to adquire and optimal action. We focus here on the subclass of RI problems with an additive cost of perception. The result of this problem is a true-value or state dependent stochastic choice rule $p_{v}(\cdot,A)\in \Delta(A\cup \{o\})$, defined as:
\[
p_{v}(\cdot,A)=\argmax_{p}\sum_{a\in A\cup \{o\}}p_v(a,A)v_a\mu(v_a)-\kappa(p_v(\cdot,A),\mu).
\]
For the specification of $\kappa$, we focus on the generalized entropy proposed in \citet{fosgerau2017discrete}, which generalizes widely used entropic cost. \citet{fosgerau2017discrete} shows that this state-dependent stochastic choice is observationally equivalent to an additive random utility choice rule conditional on the support. That is, if $p_v(\cdot,A)\in\Delta(A\cup \{o\})$ (positive probability of choice), then $p_v(\cdot,A)$ is a random utility rule. Even when the underlying utility is fixed (and equal to $v$ without loss of generality), there is randomness in choice due to costly information acquisition. The state-dependent stochastic choice only differs from $\rumo$ when there are items in the choice set that are never chosen. Therefore, the RI DM is compatible with \textbf{deterministic} consideration sets. However, in our experiment we do not observe any element chosen with zero probability. In fact, the lowest probability of choice is $6$ percent across all alternatives in $X\cup \{o\}$ and across all choice sets. 
\par
We have to aggregate across states to derive testable implications for the representative RI DM for our dataset. This is because in our setup, the experimenter does not know ex-ante the true value of alternatives. Preferences over lotteries (when there is not first-order stochastic dominance ordering among them) is unknown before choice. This is an important difference between our experiment and RI experimental literature, since they generally rely on enhanced datasets. We focus on collecting datasets that replicate standard stochastic choice data. 
\par
Using the fact that the sum of random utility rules is also a random utility rule, we notice that the marginal probability of choosing across different states is just the sum over the likelihood of this states (or the distribution of the true preferences). Then, if $P$ admits a representative RI DM: 
\[
 p(a,A)=\sum_{v\in V} p_{v}(a,A)\rho(v),
\]
where $\rho\in \Delta(V)$ is the objective probability of the unobserved states. 
\begin{lemma}
If $p_v(a,A)>0$ for all $a\in A\cup \{o\}$, and all $v\in V$, it follows that if $P$ admits a representative RI DM then, $P$ also admits a $\rumo$ representation. 
\end{lemma}
The proof of this lemma follows from \citet{fosgerau2017discrete} and \citet{ABDsatisficing16} that showed that the weighted sum of $\rumo$ is also $\rumo$. The case in which one allows heterogeneity in discrete consideration sets, induced by RI, is difficult and left for future research.  

\section{Experiment}\label{sectionapp:section_experiment}

\subsection{Sample}\label{subsectionapp:sample}
The primitive for the considered models is the estimated stochastic choice rule $\hat{p}_f(a,A)$ for $f\in\{\text{H, M, L}\}$. Therefore, for a fixed level of the cost $f$, the minimal required sample size was calculated to be proportional to the cardinality of the choice set. To maximize the number of observations for a given set of individuals, some individuals faced two decision tasks. To prevent possible learning, these subjects faced disjoint choice sets. That is, every subject faced either the full choice set $X\cup\{o\}$ or two choice sets that only had the outside option in common. Therefore, because of random assignment, in our experiment 
\begin{enumerate}
	\item  171 subjects faced only the whole choice set (the targeted number is 180);
	\item  757 subjects faced pairs of disjoint choice sets: the set of size 4 and the set of size 1 (the targeted number is 750);
	\item 1207 subjects faced pairs of disjoint choice sets: the set of size 3 and the set of size 2 (the targeted number is 1200).
\end{enumerate}

This implies a total of 2135 subjects (the targeted number is 2130) for a total 4099 choices (the targeted number is 4080). Additionally, demographic data and preferences over binary comparison of lotteries were asked and incentivized. The effective number of observations per alternative/choice set/cost is summarized in Table~\ref{observations}.

			\begin{table}	\caption[Observations]{\textsc{Average number of observations per alternative/choice set}}  \label{observations} 
		\centering
			\begin{tabular}{@{\extracolsep{10pt}} |c | c | c | | c | c | c | @{}}
				\cline{1-6}
&&&&&\\
\textsc{Choice set}	 & ~~~$N$~~~ & ~~~$N/\abs{A}$~~~ &\textsc{Choice set}	 & ~~~$N$~~~ & ~~~$N/\abs{A}$~~~ \\ [1ex]	 \cline{1-6}
\textbf{o12345}& 171 & 28.50  & \textbf{o12}& 131 & 43.67   \\ [.5ex]	 \cline{1-3}
\textbf{o2345}& 155 & 31.00 &\textbf{o13}& 118 & 39.33  \\[.5ex]
\textbf{o1345}& 154 & 30.80 &\textbf{o14}& 125 &41.67    \\	 [.5ex]
\textbf{o1245}& 149 & 29.80&\textbf{o15}& 116 & 38.67 \\[.5ex]		
\textbf{o1235}& 156 & 31.20 &\textbf{o23}& 112 & 37.33 \\[.5ex]		
\textbf{o1234}& 143 & 28.60  &\textbf{o24}& 123 & 41.00  \\ [.5ex]	\cline{1-3}
\textbf{o345}  & 131 &32.75 &\textbf{o25}& 120 & 40.00 \\	  [.5ex]
\textbf{o245}& 118 & 29.50 &\textbf{o34}& 121 & 40.33  \\	[.5ex]
\textbf{o235}& 125 & 31.25  &\textbf{o35}& 122 & 40.67  \\[.5ex]	
\textbf{o234}& 116 & 29.00 &\textbf{o45} & 119 & 39.67 \\[.5ex]	\cline{4-6}	
\textbf{o145}& 112 & 28.00 &\textbf{o1}& 155 & 77.50   \\[.5ex]	
\textbf{o135}& 123 & 30.75  &\textbf{o2}& 154 & 77.00  \\[.5ex]	
\textbf{o134}& 120 & 30.00  &\textbf{o3}& 149 & 74.50 \\[.5ex]	
\textbf{o125}& 121 & 30.25 &\textbf{o4}& 156 & 78.00  \\[.5ex]	
\textbf{o124}& 122 & 30.50 &\textbf{o5}& 143 & 71.50 \\	[.5ex]

\textbf{o123}& 119 & 29.75 &&&\\[.5ex]	 \cline{1-6}
\end{tabular}
		\end{table}

\section{Performance of the Test}\label{app: MC simulations}

In this section we study the performance of our test in terms of statistical power. We are going to test the null  hypothesis of $\la$-$\hrc$ when the true choice process presents choice overload.  We consider behavior arising from a mixed population. A fraction $\lambda\in[0,1]$ of the population is consistent with $\mm\text{-}\hrc$ with $\gamma(x)=1/2$ for all $x\in X$ and preferences consistent with expected utility maximization. The remaining fraction, $\lambda$, follows simple heuristics such the DM chooses outside option with probability proportional to the cardinality of the set. If she decides to pay attention to the menu, then she chooses uniformly at random from it.  The process is then consistent with the following stochastic choice rule
\begin{equation}
p(a,A)=\lambda p^{\mm\text{-}\hrc}(a,A)+(1-\lambda) p^{CO}(a,A)
    \label{eq.process}
\end{equation}
where $p^{\mm\text{-}\hrc}(a,A)$ is consistent with $\mm\text{-}\hrc$ with $(m^{\mm},\pi)$ and $\pi(\succ)=1/120$ for all $\succ\in R(X)$; and 
\[
p^{CO}(o,A)=\frac{\abs{A}+1}{\abs{X}+1} \text{ and }p^{CO}(a,A)=\frac{1-p^{CO}(o,A)}{\abs{A}}.
\]
The assumed process implies that a fraction $1-\lambda$ of the population exhibits choice overload.\footnote{This process may intuitively arise when a DM that faced with a choice set only knows the size of the choice set and the alternatives in the grand set $X$. Knowing about the alternatives implies paying a cost $c$ per alternative. Assume that preferences over information are modelled by a willingness to pay attention variable, $w$, that is distributed uniformly in $[0,1]$. Then, given a choice set realization, after knowing $\abs{A}$ DM $i$ decides to pay attention to choice set $A$ if $w_i>\abs{A}\times c$ . This implies that the DM pays attention and decide in the interior of the set with probability $\sum_{a\in A}p(a,A)=1-c\abs{A}$; and $p(o,A)=c \abs{A}$. }

As the proportion of the population that exhibits choice overload increases so should increase the probability of rejecting the null that population behavior is generated by $\rhrc$. On the other extreme, when $\lambda=1$ we should not reject the model. In particular, for any $\lambda<32/39$ the process defined by equation~(\ref{eq.process}) exhibits choice overload. However,  for high values of $\lambda$ the magnitude of this effect may not be significant to reject $\rhrc$.

Table~\ref{Table:power} presents the results for power simulations for sample size $4000$ and $\lambda\in \{0.25,0.50\}$. For 500 replications, the table displays the proportion of simulations that are rejected at the $10$ percent and $5$ percent significance levels. As expected, the fraction of rejections is bigger for smaller values of $\lambda$. For $\lambda=0.25$ the rejection rate is $100$ percent.
We observe that at comparable sample size to our experiment the mixed process is rejected with power close to 1 when the choice overload fraction of DMs is moderate. 

\begin{table}	\caption[Performance of the Test]{The table displays the proportion of rejections at the $10$ percent and $5$ percent significance levels for $\la\text{-}\hrc$. Sample size=4000. Number of MC replications=500. Number of bootstrap replications=500}  \label{Table:power} 
\centering
\begin{tabular}{ |c|c| c| @{}}
\cline{1-3}
\multicolumn{3}{|c|}{\tiny{}}\\
&\multicolumn{2}{c|}{\textsc{Significance level}}\\[1ex] 
\textsc{Process}& \multicolumn{1}{c|}{10\%}& \multicolumn{1}{c|}{5\%} \\ [1ex]
\cline{1-3}
&&\\
~~$\lambda~=0.25$~~~	&~~~~~1~~~~~&~~~~~1~~~~~\\[1ex]
~~$\lambda~=0.50$~~~ &~~~~~0.464~~~~~&~~~~~0.73~~~~~\\[1ex]
\cline{1-3}
\end{tabular}
\end{table}

\section{The  \texorpdfstring{$\rc$}{Behavioral}-model with Homogeneous Preferences\label{app: homogeneous LA}}
In our setup, given an attention rule $\rc$ (e.g., $\la$ or $\mm$) and a given frame, we can recover the underlying full consideration probabilistic choice rule $p^{\rc}_{\pi}$. Under the null hypothesis that our experimental dataset can be generated by a $\rc$-rule with a strict preference relation over $X\cup \{o\}$ (such that $o$ is the worst alternative), it must be that $p^{\rc}_{\pi}$ is degenerate (i.e., $p^{\rc}_{\pi}(a,A)\in \{0,1\}$ for \emph{all} $A$ and $a\in A$). In particular, if we reject that $p^{\rc}_{\pi}(a_1,\{a_1,a_2\})\in \{0,1\}$ for \emph{some} binary menu $\{a_1,a_2\}$ and for \emph{some} attention cost, then we have to reject the $\rc$ model with homogeneous preferences. 
\par
Given that the only model with heterogeneous preferences that was not rejected in our experiment is $\la$, to show the importance of preference heterogeneity, we tested whether the $\la$ model with homogeneous preferences can explain the data. We took menu $\{1,3\}$ and the high cost frame and computed the implied by the $\la$ model the full consideration probability $\hat{p}^{\la}_{\pi}(1,\{1,3\})$. If the data can be explained by the $\la$ model with homogeneous preferences, then $\hat{p}^{\la}_{\pi}(1,\{1,3\})$ should converge in probability to either $0$ ($3\succ 1$ with probability 1) or $1$ ($1\succ 3$ with probability 1). We tested two hypotheses: (i) ${p}_{\pi}(1,\{1,3\})=0$ and (ii) ${p}_{\pi}(1,\{1,3\})=1$. Both were rejected at the $5$ percent significance level ($\text{p-value}<10^{-3}$). As a result we can conservatively claim that the null hypothesis that our experimental dataset can be explained by the $\la$-rule with a single preference relation is rejected at least at the $5$ percent significance level. 
\section{Experiment Instructions} \label{sectionapp:experiment_instructions}

\begin{figure}[H]
		\centering
		\includegraphics[height=9cm]{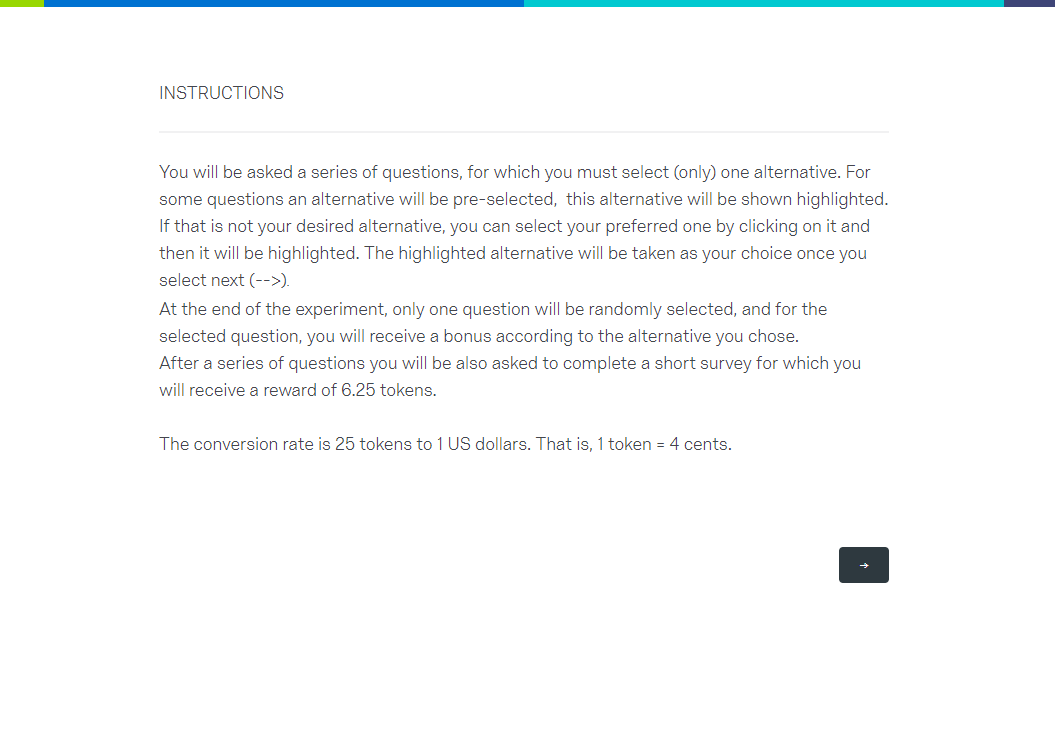}
		\caption{Instructions page 1}\label{instructions1}
	\end{figure}
\begin{figure}[H]
		\centering
		\includegraphics[height=9cm]{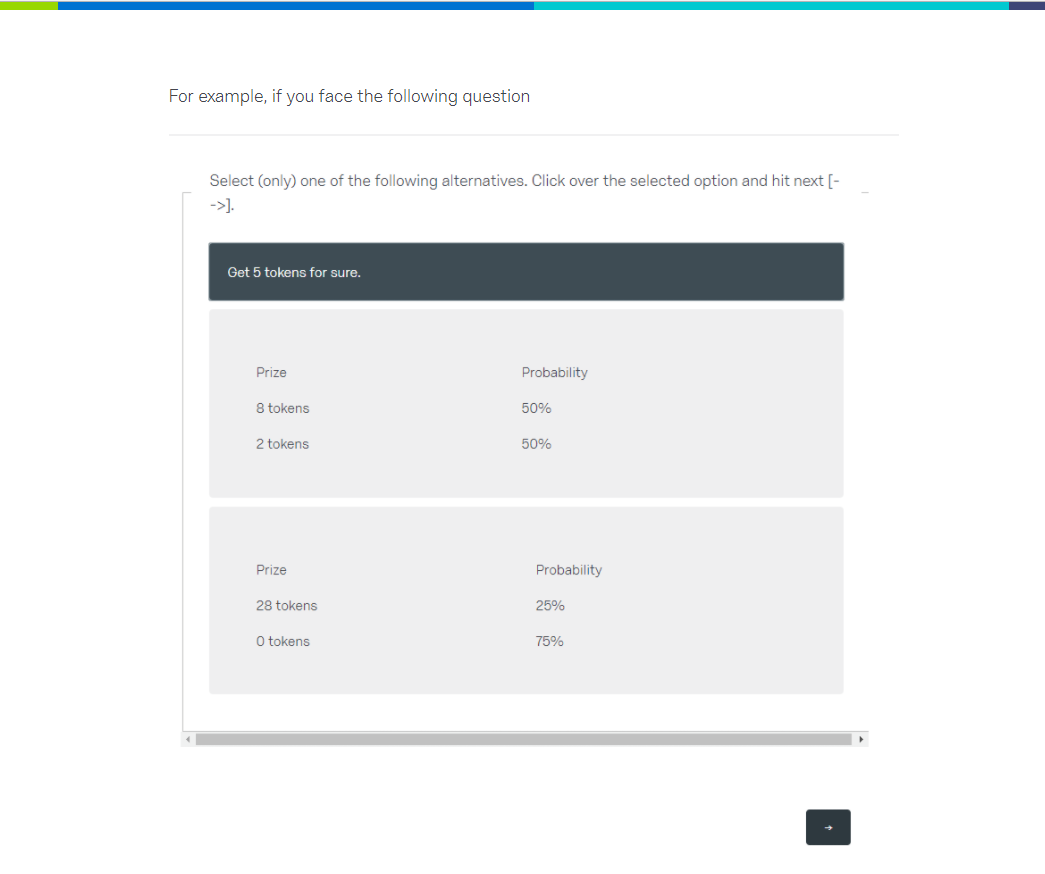}
		\caption{Instructions page 2}\label{instructions2}
	\end{figure}
	\begin{figure}[H]
		\centering
		\includegraphics[height=10cm]{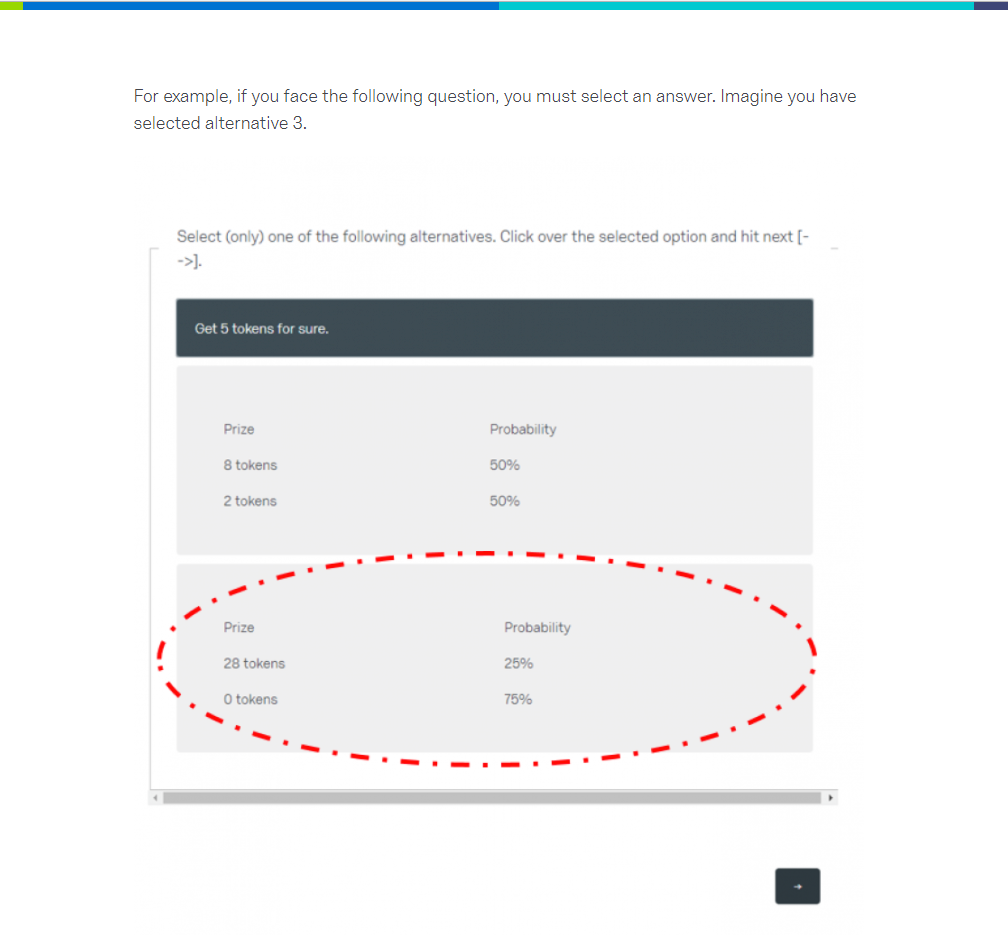}
		\caption{Instructions page 3}\label{instructions3}
	\end{figure}
	\begin{figure}[H]
		\centering
		\includegraphics[height=10cm]{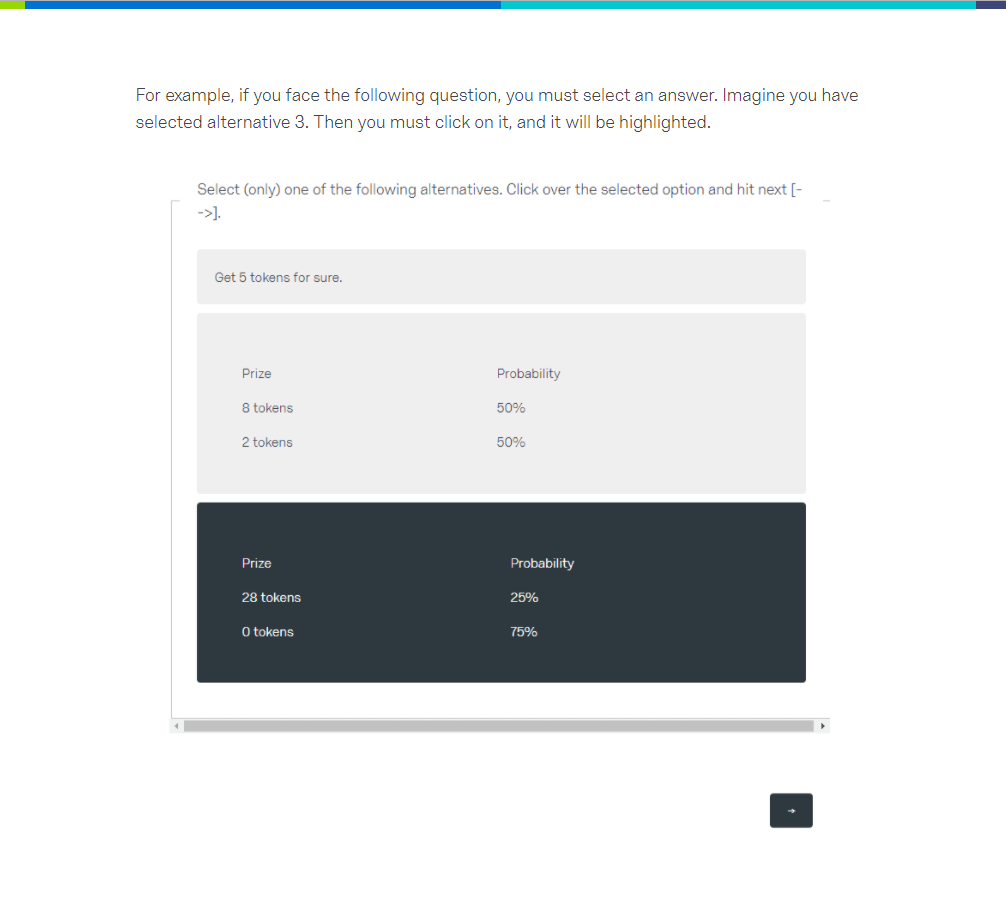}
		\caption{Instructions page 4}\label{instructions4}
	\end{figure}
	\begin{figure}[H]
		\centering
		\includegraphics[height=10cm]{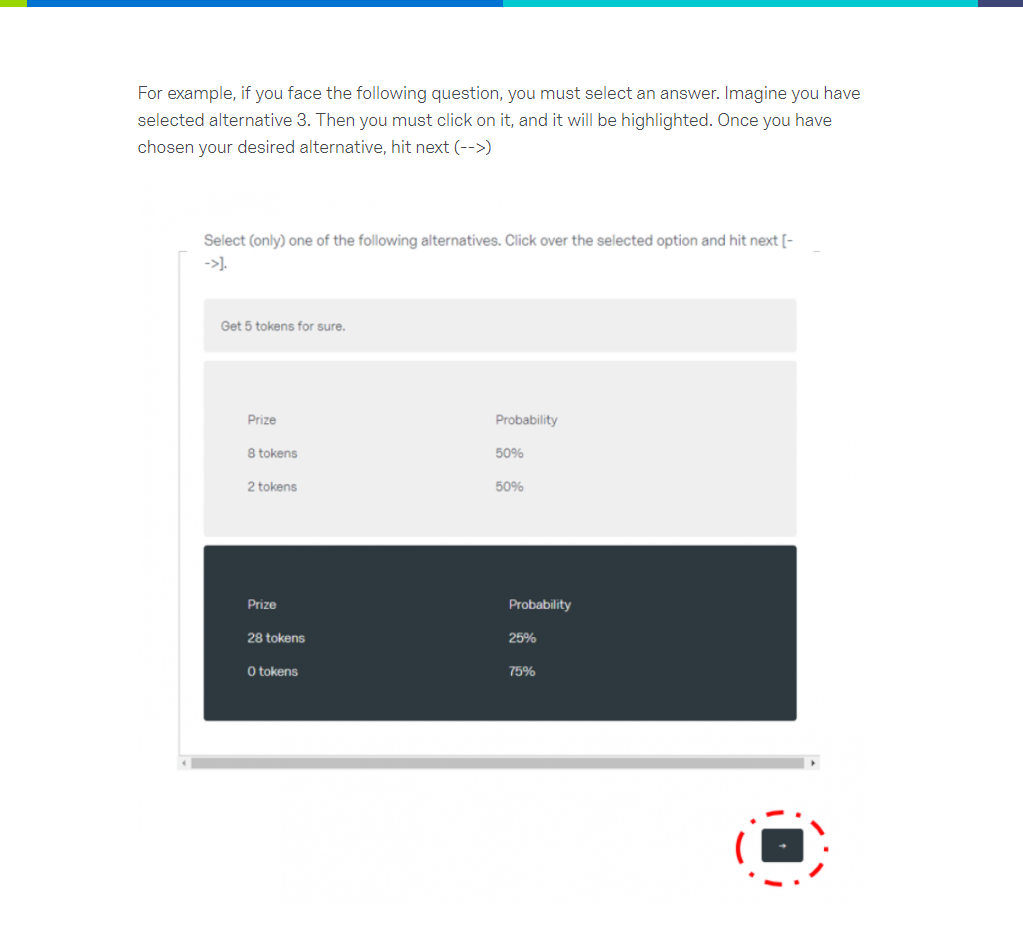}
		\caption{Instructions page 5}\label{instructions5}
	\end{figure}
	\begin{figure}[H]
		\centering
		\includegraphics[height=10cm]{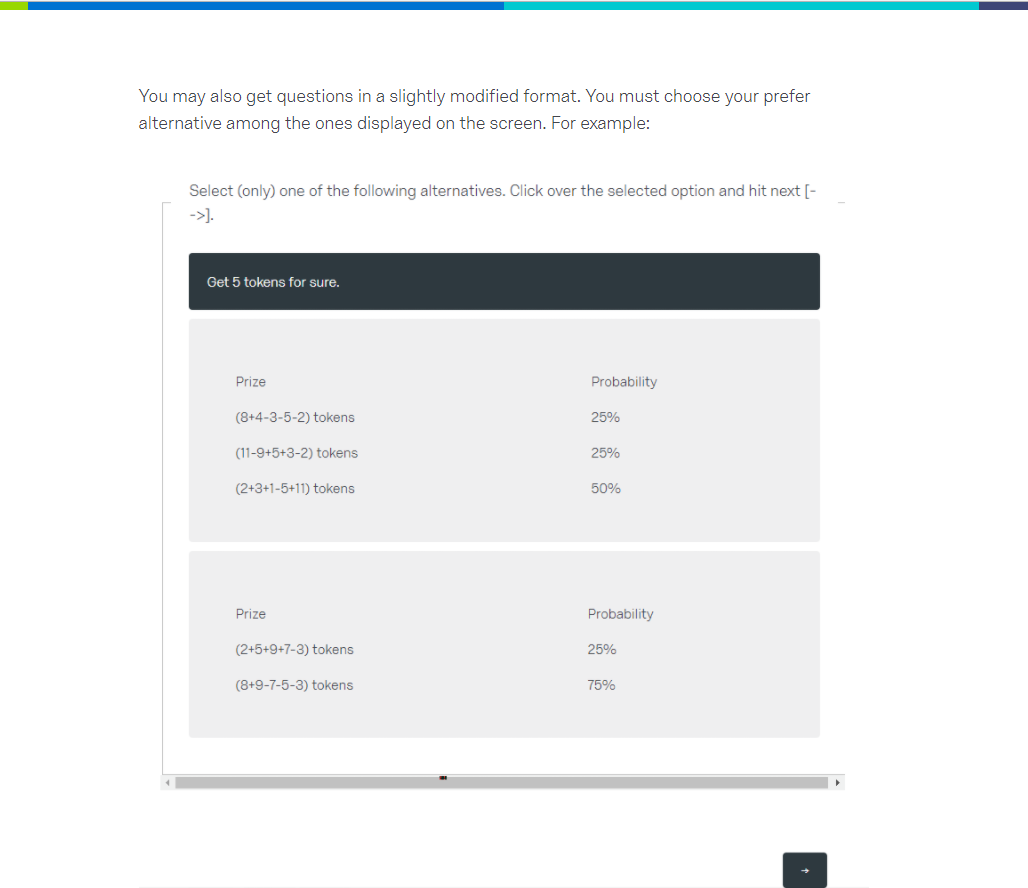}
		\caption{Instructions page 6}\label{instructions6}
	\end{figure}
	\begin{figure}[H]
		\centering
		\includegraphics[height=10cm]{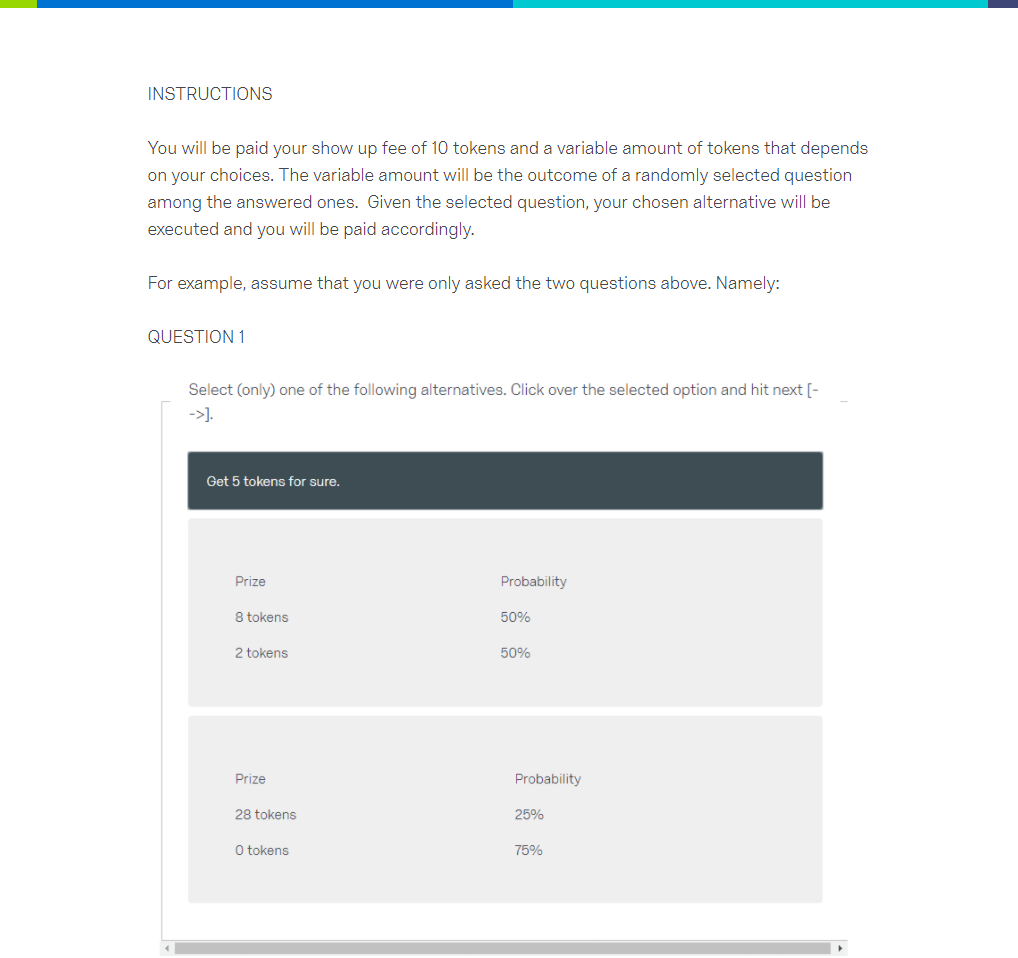}
		\caption{Instructions page 7}\label{instructions7}
	\end{figure}
	\begin{figure}[H]
		\centering
		\includegraphics[height=10cm]{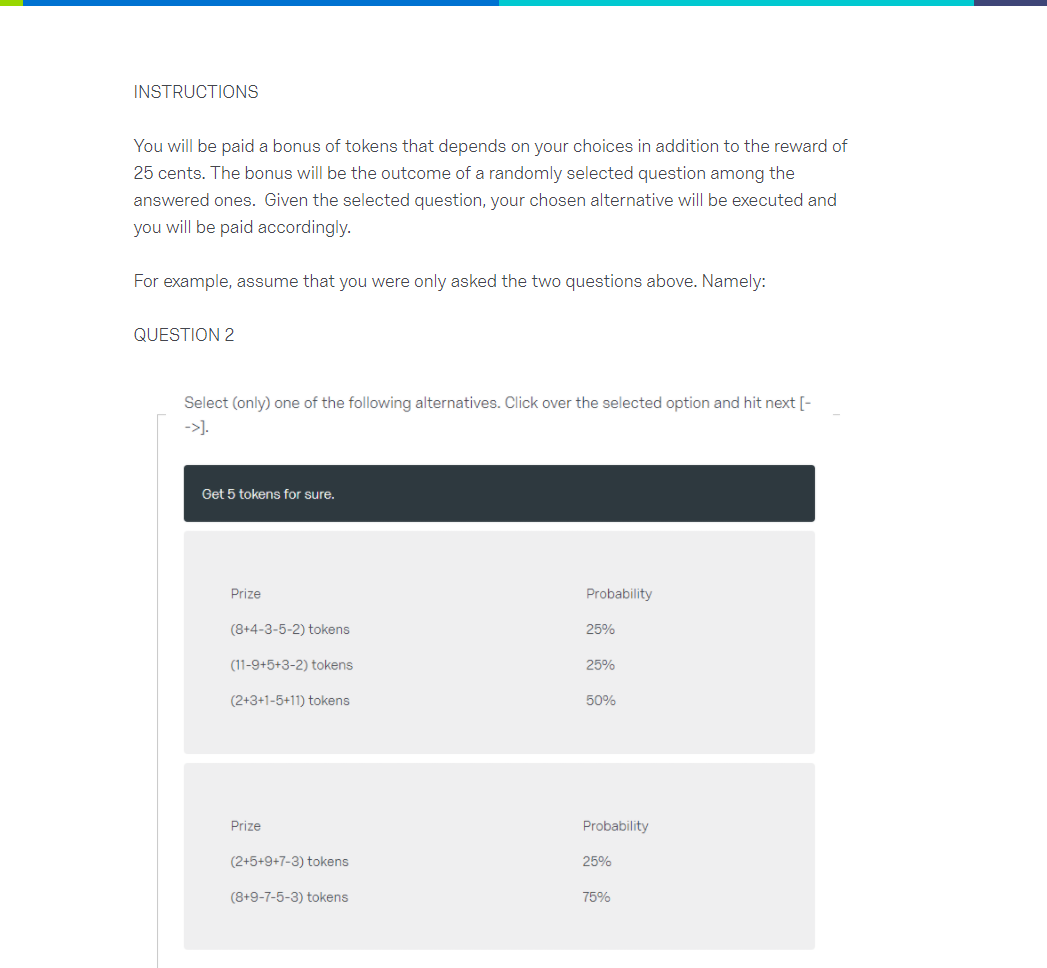}
		\caption{Instructions page 8}\label{instructions8}
	\end{figure}
	\begin{figure}[H]
		\centering
		\includegraphics[height=10cm]{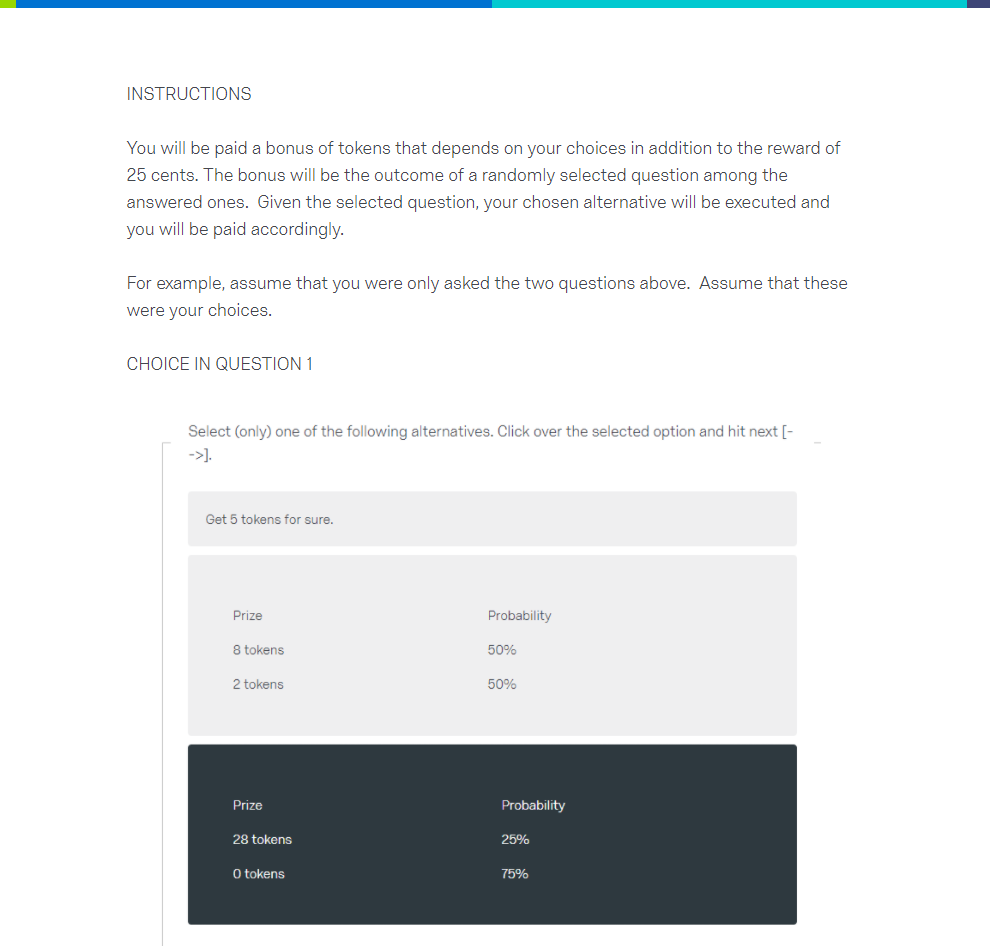}
		\caption{Instructions page 9}\label{instructions9}
	\end{figure}
	\begin{figure}[H]
		\centering
		\includegraphics[height=11cm]{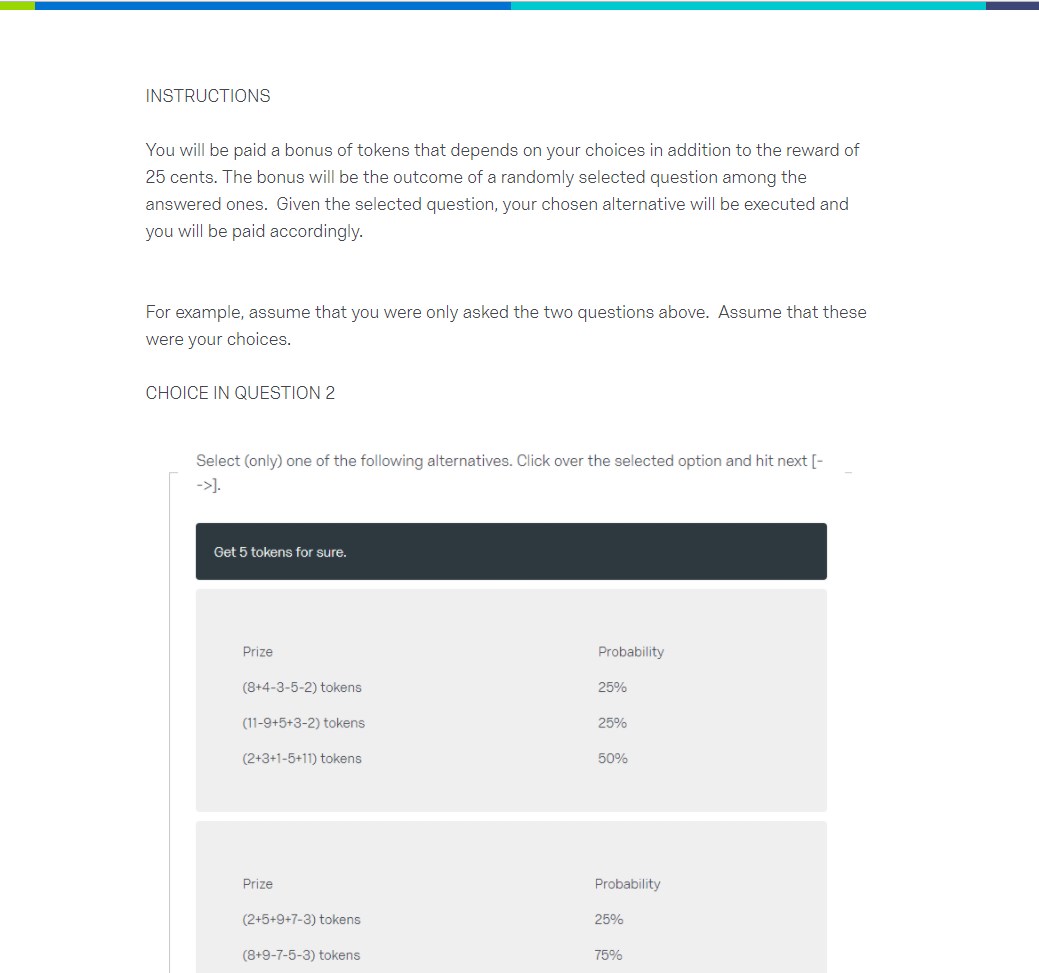}
		\caption{Instructions page 10}\label{instructions10}
	\end{figure}
	\begin{figure}[H]
		\centering
		\includegraphics[height=10cm]{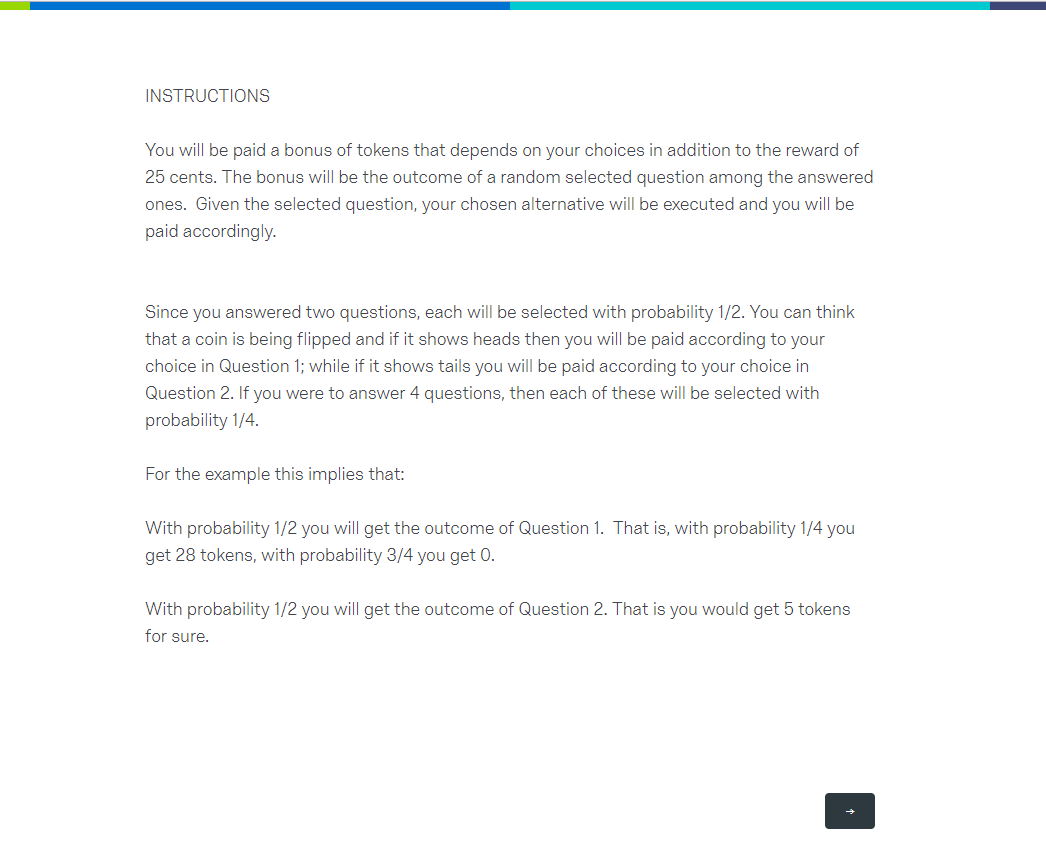}
		\caption{Instructions page 11}\label{instructions11}
	\end{figure}
	\begin{figure}[H]
		\centering
		\includegraphics[height=11cm]{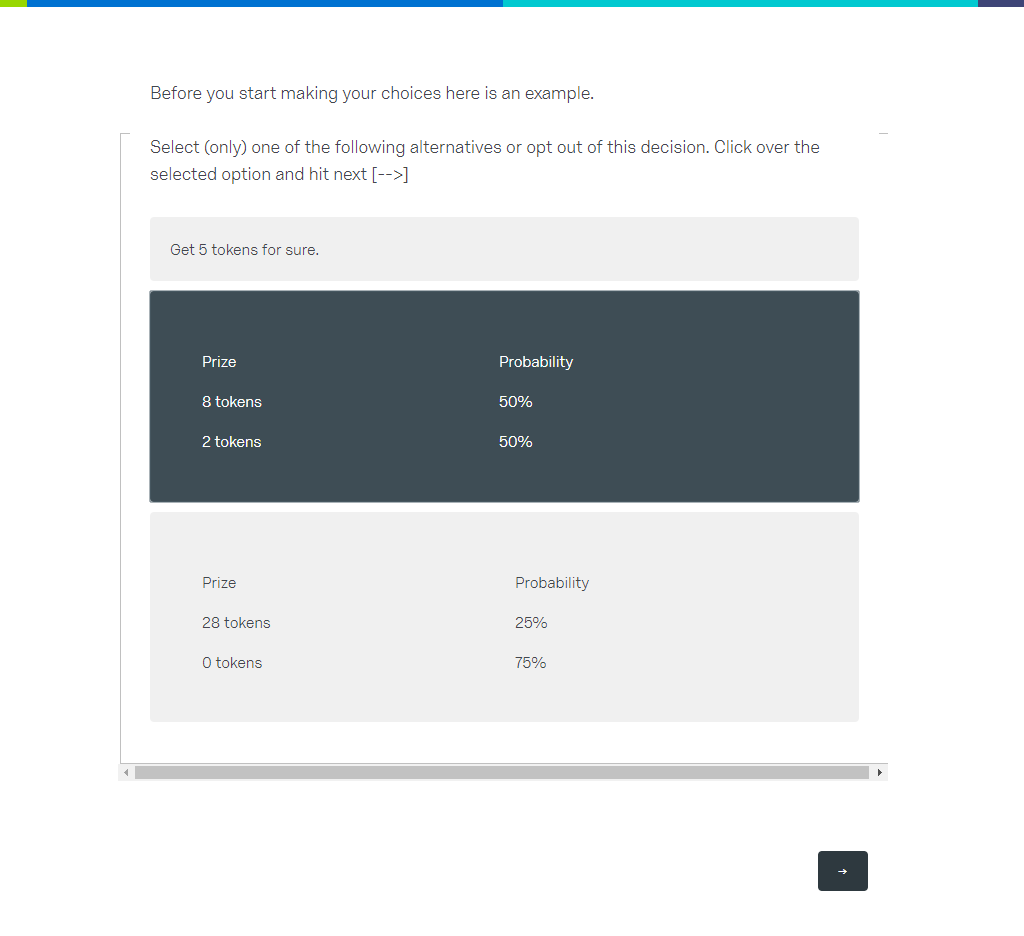}
		\caption{Instructions page 12}\label{instructions12}
	\end{figure}
	\begin{figure}[H]
		\centering
		\includegraphics[height=4.75cm]{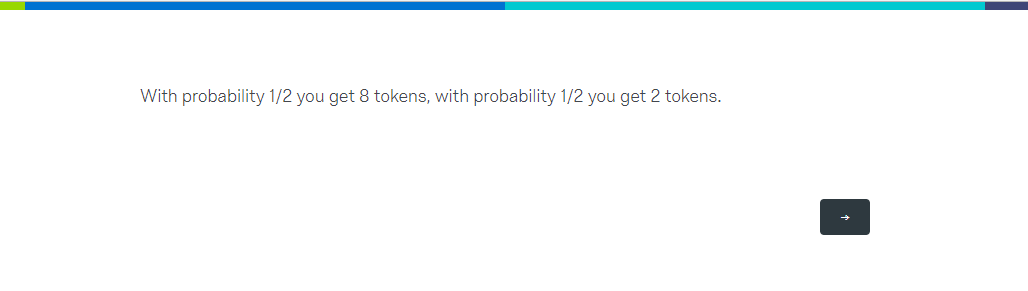}
		\caption{Instructions page 13}\label{instructions13}
	\end{figure}
	\begin{figure}[H]
		\centering
		\includegraphics[height=4.75cm]{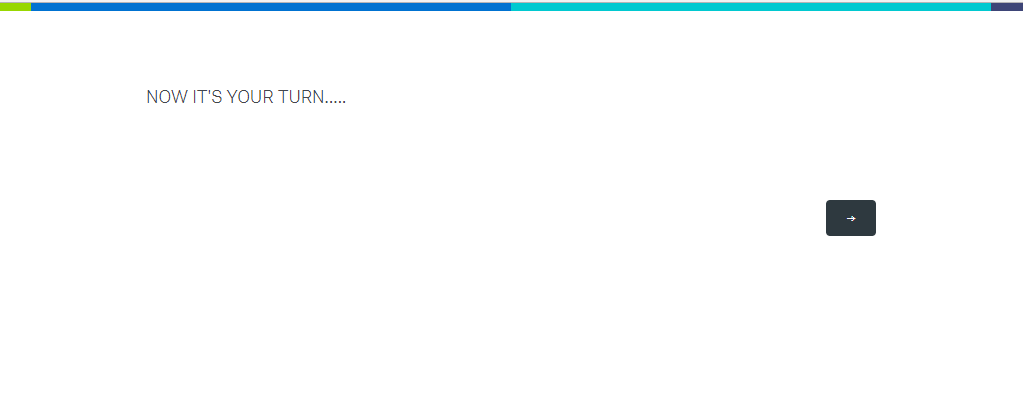}
		\caption{Instructions page 14}\label{instructions14}
	\end{figure}

\bibliographystyle{ecca}
\bibliography{fieldconsideration}

\end{document}